%% file: paper.tex
\documentclass[12pt]{article}

\input{myparentheses.tex}

\input{mycommands_shared.tex}

\input{mycommands.tex}

\input{myplots.tex}

\input{mystructure.tex}

\usepackage[
    a4paper,
    top=3cm,
    bottom=4cm,
    left=2.5cm,
    right=2.5cm,
]{geometry}

\providebool{cleanbuild}
\setbool{cleanbuild}{true}

\begin{document}

\input{tex/00_TitlePage.tex}
\clearpage

\section{Introduction}
\label{sec:introduction}
\input{tex/01_Introduction.tex}

\section{Background}
\label{sec:background}
\input{tex/02_Background.tex}

\section{Oblique projections of covariance matrices}
\label{sec:obliqueProjections}
\input{tex/03_ObliqueProjections.tex}

\section{Covariance structures}
\label{sec:covarianceStructures}
\input{tex/04_CovarianceStructures.tex}

\section{Connections and transferrable methods}
\label{sec:connections}
\input{tex/05_Methods.tex}

\section{Other connections and outlook}
\label{sec:outlook}
\input{tex/06_Outlook.tex}

\clearpage
\appendix

\section{Oblique projections details}
\label{sec:obliqueProjectionsDetails}
\input{tex/A_Projections.tex}

\section{\MP{} inverse}
\label{sec:MPInverse}

\input{tex/A_TechnicalDetails.tex}

\section{Proofs}
\label{sec:proofs}
\input{tex/A_Proofs.tex}

\section{Notation comparison}
\label{sec:notationComparison}
\input{tex/A_Notation.tex}

\clearpage
\bibliographystyle{apalike}
\bibliography{literature,literature_new,literature_comp}

\end{document}

%% file: myparentheses.tex
\newcommand{\normalleftrightnull}{\delimitershortfall 5pt \delimiterfactor 901}
\newcommand{\normalleftright}{\normalleftrightnull\aftergroup\normalleftrightnull}

\newcommand{\midleftrightnull}{\delimitershortfall 15pt \delimiterfactor 700}
\newcommand{\midleftright}{\midleftrightnull\aftergroup\midleftrightnull}

\newcommand{\smallleftrightnull}{\delimitershortfall 30pt \delimiterfactor 400}
\newcommand{\smallleftright}{\smallleftrightnull\aftergroup\smallleftrightnull}

\newcommand{\noleftrightnull}{\delimitershortfall 500pt \delimiterfactor 0}
\newcommand{\noleftright}{\noleftrightnull\aftergroup\noleftrightnull}

\midleftright %

\newcommand{\myleft}{\mathopen{}\mathclose\bgroup\left}
\newcommand{\myright}{\aftergroup\egroup\right}

\newcommand{\lr}[1]{\myleft(#1\myright)}

\newcommand{\set}[1]{\myleft\{#1\myright\}}

\newcommand{\setm}[2]{\myleft\{#1 \;\middle\vert\; #2\myright\}}

\newcommand{\abs}[1]{\left|{#1}\right|}

\newcommand{\brackets}[1]{\myleft\lbrack#1\myright\rbrack}

\newcommand{\halfopen}[1]{\myleft\lbrack#1\myright)}

%% file: mycommands_shared.tex
\usepackage{amsmath, amssymb} %
\usepackage{mathtools} %
\usepackage{bbm} %

\newcommand{\Acal}{\mathcal{A}}

\newcommand{\Dcal}{\mathcal{D}}

\newcommand{\Mcal}{\mathcal{M}}

\newcommand{\Pcal}{\mathcal{P}}

\newcommand{\Scal}{\mathcal{S}}

\newcommand{\expmeasure}{\Lambda}
\newcommand{\expdens}{\lambda}

\newcommand{\inftyvec}{\boldsymbol{\infty}}
\newcommand{\zerovec}{\mathbf{0}}
\newcommand{\onevec}{\mathbf{1}}

\newcommand{\indic}[1]{\mathbbm{1}_{#1}}

\newcommand{\kev}[1]{\mathbf{e}_{#1}}
\newcommand{\ek}[1][k]{\kev{#1}}

\newcommand{\Var}{\mathrm{Var}}

\newcommand{\Cov}{\mathrm{Cov}}
\providecommand{\P}{}
\renewcommand{\P}[1][]{\mathbb{P}_{#1}}
\newcommand{\E}[1][]{\mathbb{E}_{#1}}

\newcommand{\normal}{\mathcal{N}}

\newcommand{\N}{\mathbb{N}}
\newcommand{\R}{\mathbb{R}}
\newcommand{\Rdd}[1][d]{\R^{#1 \times #1}}
\newcommand{\Rd}[1][d]{\R^{#1}}

\newcommand{\slr}[1]{^{\lr{#1}}}

\newcommand{\T}{^\top\!}
\newcommand{\inv}{^{-1}}

\newcommand{\pinv}{^+}
\newcommand{\pinvT}{^{+\top}}
\newcommand{\orth}{^{\perp}}
\newcommand{\ocompl}[1]{#1\orth}
\newcommand{\sonevec}{^{\onevec}}

\newcommand{\without}[1]{\setminus\set{#1}}
\newcommand{\restrictto}[1]{\big\vert_{#1}}

\newcommand{\excondind}[3]{#1 \perp_e #2 \mid #3}
\newcommand{\simiid}{\overset{\mbox{\tiny i.i.d.}}{\sim}}
\newcommand{\deq}{\stackrel{d}{=}}

\newcommand{\pdet}[1]{\abs{#1}_+}

\newcommand{\trace}{\mathrm{tr}}
\newcommand{\laspan}{\mathrm{span}}
\newcommand{\image}{\mathrm{Im}}
\newcommand{\kernel}{\mathrm{ker}}
\newcommand{\diag}{\mathrm{diag}}

\newcommand{\Id}{\mathrm{I}}

\newcommand{\vario}{\gamma}

\newcommand{\limit}[2][\infty]{\lim_{#2\rightarrow #1}}
\newcommand{\half}{\tfrac{1}{2}}

\newcommand{\myFunctionDefinitionInline}[5]{
  #1: \; & #2 &\longrightarrow #3
  , \quad
  #4 \longmapsto #5
}

\newcommand{\CAP}{\expandafter\MakeUppercase}
\newcommand{\HR}{Hüsler--Reiss}

\newcommand{\MP}{Moore--Penrose}

\newcommand{\multGP}{multivariate generalized Pareto}
\newcommand{\MultGP}{Multivariate generalized Pareto}

\newcommand{\psd}{positive semidefinite}

\newcommand{\cnd}{conditionally negative definite}

\newcommand{\spieceasi}{\texttt{SPIEC-EASI}}
\newcommand{\gcoda}{\texttt{gCoda}}
\newcommand{\sparcc}{\texttt{SparCC}}
\newcommand{\eglearn}{\texttt{eglearn}}
\newcommand{\cglearn}{\texttt{cglearn}}
\newcommand{\cclasso}{\texttt{CCLasso}}
\newcommand{\easycoda}{\texttt{easyCODA}}
\newcommand{\gemas}{\texttt{gemas}}

\newcommand{\robCompositions}{\texttt{robCompositions}}

%% file: mycommands.tex
\usepackage{amsmath, amssymb} %
\usepackage{mathtools} %
\usepackage{bbm} %

\newcommand{\setpsd}[1]{\Pcal_{#1}}

\newcommand{\setpsdt}[1]{\Pcal^0_{#1}}

\newcommand{\setpsdany}{\Scal}

\newcommand{\setMv}[1]{\Mcal_{#1}}

\newcommand{\setMvt}[1]{\Mcal^0_{#1}}

\newcommand{\setcnd}{\Dcal}

\newcommand{\PX}[1]{P_{#1}}

\newcommand{\PXY}[2]{P_{#1#2}}
\newcommand{\PXYo}[2]{P_{#1#2\orth}}
\newcommand{\Pxyo}[2]{P_{#1#2\orth}}

\newcommand{\PUV}{\PXY{U}{V}}

\newcommand{\PUUo}{\PXYo{U}{U}}
\newcommand{\PVVo}{\PXYo{V}{V}}

\newcommand{\PUVo}{\PXYo{U}{V}}
\newcommand{\PUWo}{\PXYo{U}{W}}
\newcommand{\PVUo}{\PXYo{V}{U}}

\newcommand{\PWUo}{\PXYo{W}{U}}

\newcommand{\pXYo}[2]{\pi_{#1#2\orth}}
\newcommand{\pxyo}[2]{\pi_{#1#2\orth}}

\newcommand{\pUUo}{\pXYo{U}{U}}
\newcommand{\pVVo}{\pXYo{V}{V}}

\newcommand{\pUVo}{\pXYo{U}{V}}
\newcommand{\pUWo}{\pXYo{U}{W}}
\newcommand{\pVUo}{\pXYo{V}{U}}

\newcommand{\puvo}{\pxyo{u}{v}}

\newcommand{\pvvo}{\pxyo{v}{v}}
\newcommand{\pevo}{\pxyo{\onevec}{v}}
\newcommand{\pveo}{\pxyo{v}{\onevec}}
\newcommand{\peeo}{\pxyo{\onevec}{\onevec}}
\newcommand{\pexo}[1]{\pxyo{\onevec}{#1}}
\newcommand{\pxeo}[1]{\pxyo{#1}{\onevec}}
\newcommand{\pxxo}[1]{\pxyo{#1}{#1}}

\newcommand{\Puvo}{\Pxyo{u}{v}}
\newcommand{\Puwo}{\Pxyo{u}{w}}
\newcommand{\Pevo}{\Pxyo{\onevec}{v}}
\newcommand{\Pveo}{\Pxyo{v}{\onevec}}
\newcommand{\Peeo}{\Pxyo{\onevec}{\onevec}}

\newcommand{\Pexo}[1]{\Pxyo{\onevec}{#1}}

\newcommand{\LR}{\mathrm{LR}}
\newcommand{\ALR}[1]{\mathrm{ALR}_#1}
\newcommand{\CLR}{\mathrm{CLR}}

\newcommand{\cl}{\mathrm{C}}
\newcommand{\lcl}{\tilde{\mathrm{C}}}
\newcommand{\simplex}{\Delta}

\newcommand{\ai}[1]{\underline{#1}}

%% file: myplots.tex
\usepackage{etoolbox} %

\providecommand{\providesetbool}[2]{%
  \ifcsundef{if#1}{%
    \newbool{#1}%
    \setbool{#1}{#2}%
  }{}%
}

\usepackage{tikz}
\usetikzlibrary{matrix}
\usetikzlibrary{external}
\usetikzlibrary{arrows.meta}

\providesetbool{externalizetikz}{false}

\ifbool{externalizetikz}{%
  \tikzsetexternalprefix{externalized/}
  \tikzexternalize[]
}{}

\usepackage{float} %

\providesetbool{useXdiagrams}{true}

\providebool{ignoretikz}
\setbool{ignoretikz}{false}

\newcommand{\redtt}[1]{{\color{red}\ttfamily #1}}

\newcommand{\inputtikzwithfilename}[1]{%
  \ifbool{ignoretikz}{%
    \redtt{Ignoring tikz: \detokenize{#1}}%
  }{%
    \IfFileExists{#1}{%
      \tikzsetnextfilename{#1}%
      \input{#1}%
    }{%
      \redtt{File not found: \detokenize{#1}}%
    }%
  }%
}

\newcommand{\inputtikz}[1]{%
  \input{tikz/#1}%
}
\newcommand{\inputtikzR}[1]{
  \input{tikzR/#1}%
}

%% file: mystructure.tex
\usepackage{etoolbox} %
\providebool{cleanbuild}

\usepackage{placeins} %

\usepackage{amsthm}

\usepackage{xr-hyper}
\usepackage[hidelinks]{hyperref}
\usepackage[capitalize,nameinlink]{cleveref} %

\crefformat{equation}{#2(#1)#3}
\crefrangeformat{equation}{#3(#1)#4--#5(#2)#6}
\crefmultiformat{equation}{#2(#1)#3}{ and~#2(#1)#3}{, #2(#1)#3}{, and~#2(#1)#3}
\Crefrangeformat{equation}{#3(#1)#4--#5(#2)#6}
\Crefmultiformat{equation}{#2(#1)#3}{ and~#2(#1)#3}{, #2(#1)#3}{, and~#2(#1)#3}

\crefname{algocf}{Algorithm}{Algorithms}
\Crefname{algocf}{Algorithm}{Algorithms}

\AddToHook{cmd/appendix/before}{\crefalias{section}{appendix}}

\crefformat{figure}{#2Figure~#1#3}

\usepackage{aliascnt}
\usepackage{xparse}
\let\oldtheorem\newtheorem
\RenewDocumentCommand{\newtheorem}{m o m o}{%
    \IfValueTF{#2}{%
        \IfNoValueTF{#4}{%
            \newaliascnt{#1}{#2}%
            \oldtheorem{#1}[#1]{#3}%
            \aliascntresetthe{#1}%
        }{%
            \oldtheorem{#1}[#2]{#3}[#4]%
        }%
    }{%
        \IfNoValueTF{#4}{%
            \oldtheorem{#1}{#3}%
        }{%
            \oldtheorem{#1}{#3}[#4]%
        }%
    }%
}

\makeatletter
\newcommand{\DeclareTodoCref}[2]{%

  \newcounter{todoref@#1}%

  \expandafter\renewcommand\csname theHtodoref@#1\endcsname{%
    todo.#1.\arabic{todoref@#1}%
  }%

  \crefformat{todoref@#1}{\todo{#2}}%
  \Crefformat{todoref@#1}{\todo{#2}}%

  \AtBeginDocument{%
    \refstepcounter{todoref@#1}%
    \label{#1}%
  }%
}
\makeatother

\DeclareTodoCref{todo}{ref}
\DeclareTodoCref{sec:todo}{ref section}
\DeclareTodoCref{result:todo}{ref result}
\DeclareTodoCref{lemma:todo}{ref result}
\DeclareTodoCref{theorem:todo}{ref result}
\DeclareTodoCref{fig:todo}{ref figure}
\DeclareTodoCref{table:todo}{ref table}

\usepackage[sort]{natbib}

\makeatletter \renewcommand\hyper@natlinkbreak[2]{#1} \makeatother

\numberwithin{equation}{section}

\theoremstyle{plain}
\newtheorem{theorem}{Theorem}[section]
\newtheorem{proposition}[theorem]{Proposition}
\newtheorem{lemma}[theorem]{Lemma}
\newtheorem{corollary}[theorem]{Corollary}
\newtheorem{definition}[theorem]{Definition}

\theoremstyle{definition}

\theoremstyle{remark}
\newtheorem{remark}[theorem]{Remark}

\newcommand{\provenlabel}[1]{%
  \label{#1}%
  \ifbool{cleanbuild}{}{%
    \makeatletter%
    \ifcsname r@proof:#1\endcsname%
      \hyperlink{proof:#1}{(See Proof.)}%
    \fi%
    \makeatother%
  }%
}

\newcommand{\proofref}[1]{%
  \label{proof:#1}%
  Proof of \hypertarget{proof:#1}{\cref{#1}}%
}

\usepackage{enumitem}
\setlist[enumerate,1]{label={(\arabic*)}}

\RequirePackage{algorithm}%
\RequirePackage{algorithmicx}%
\RequirePackage{algpseudocode}%
\RequirePackage{listings}%

%% file: tex/00_TitlePage.tex
\title{An Explicit Link between Extreme Value Theory and Compositional Data Analysis}

\author{%
    Manuel Hentschel, University of Geneva
    \and
    Sebastian Engelke, University of Geneva
}

\date{}

\maketitle

\begin{abstract}
    \noindent%
    \input{tex/00_Abstract.tex}
\end{abstract}

%% file: tex/00_Abstract.tex
Extreme value theory and compositional data analysis both study settings
where relative information plays a central role.
In multivariate extreme value theory,
threshold exceedance limits satisfy homogeneity properties that separate
the radial size of an extreme event from its relative profile.
In compositional data analysis,
positive vectors are analysed up to multiplicative scale,
and inference is based on ratios or log-ratios between components.
Consequently,
both fields have developed several covariance and dependence representations
of the underlying relative structure.
In the Hüsler--Reiss model for extremes,
these include variogram, covariance, and precision parametrizations.
In compositional data analysis,
analogous representations arise from pairwise log-ratios,
centred log-ratios,
and additive log-ratios.
We establish an explicit link between the two fields 
that relates these different representations by a small set of simple transformations,
including oblique projections, Hüsler--Reiss inverses, and the variogram map.
From a methodological perspective, leveraging this algebraic connection enables the transfer of statistical approaches from one field to the other. For instance,
we introduce intrinsic logistic-normal graphical models for compositional data,
which are based on Hüsler--Reiss graphical models for extremes.
Conversely,
we explore how dimensionality reduction methods from compositional data analysis
can be applied to the analysis of multivariate extremes.

%% file: tex/01_Introduction.tex
Multivariate extreme value theory provides models for the joint behaviour
of rare events,
such as extreme weather events \citep{PascheLamEngelke2026},
financial crises \citep{poonRockingerTawn2004},
or large insurance claims \citep{beirlantEtAl2004_chapter}.
A common approach in extreme value theory is to study threshold exceedances through multivariate
generalized Pareto distributions \citep{rootzen2006},
which arise as the limiting distributions of rescaled threshold exceedances.
One way to describe
the dependence structure of a multivariate generalized Pareto distribution $Y = (Y_1, \dots ,Y_d)$ in $\mathbb R^d$
is by the extremal variogram,
defined as the matrix of variances of pairwise differences between entries of the
$v$-extremal functions \citep{hentschel2025},
\begin{equation}
    \label{eq:intro_extremal_variogram}
    \Gamma^v_{ij}
    =
    \Var(Y_i - Y_j | v\T Y > 0), \quad i,j = 1,\dots, d
    .
\end{equation}
For the \HR{} model \citep{hueslerReiss1989},
this variogram is independent of the choice of $v$
and fully parametrizes the distribution.
The same matrix determines several covariance and precision representations,
including the covariance matrices of extremal functions
and the singular precision matrix used in \HR{} graphical models
\citep{engelke2020,hentschel2023}.

Compositional data analysis provides statistical methods
for data that describe the relative contributions of several parts to a whole.
Example applications include chemical compositions in geology \citep{egozcueEtAl2024},
species abundances in ecological samples \citep{BillheimerEtAl2001},
and demographic proportions in population studies \citep{LloydPawlowskyEgozcue2012}.
Since only relative information is meaningful in these datasets,
each observation is commonly divided by its total,
yielding a vector of proportions that lives on the simplex,
to which many standard modelling and inference methods for multivariate data are inapplicable.
A common solution to this problem is to work
with log-ratio transformations of the data
\citep{aitchison_statistical_1986,greenacre_compositional_2021},
and the covariance structures of the transformed data.
A central dependence measure in this setting is the compositional variation array,
defined as the matrix of pairwise variances of log-ratios,
\begin{equation}
    \label{eq:intro_variation_array}
    \Gamma_{ij}
    =
    \Var\{\log(x_i / x_j)\}, \quad i,j = 1,\dots, d
    ,
\end{equation}
where $x = (x_1,\dots, x_d)$ is a random vector of proportions, supported on the simplex.
Other important covariance structures include the covariance matrix of centred log-ratios,
and the covariance matrices of additive log-ratio coordinates \citep{aitchison_statistical_1986}.
These objects are again different representations of the same relative
covariance structure.
At a formal level, this is closely related to the covariance structures
appearing in \HR{} extreme value theory:
log-ratio transformed data in compositional data analysis
correspond to extremal functions in extreme value theory,
pairwise log-ratios correspond to differences between coordinates of extremal functions,
and the compositional variation array to the extremal variogram.

Our paper makes this connection precise by establishing a formal link between extreme value theory and compositional data analysis. Leveraging this connection enables the transfer of statistical models and methods between the two fields.
As the mathematical backbone,
in \cref{sec:obliqueProjections} we develop a framework to describe and relate the
covariance and precision matrices of
random vectors which are supported on general subspaces.
In \cref{sec:covarianceStructures},
we specialize this framework to projections whose kernel is
spanned by the one-vector,
and prove a variogram representation result,
relating negative definite variograms to covariance matrices on
hyperplanes complementary to the one-vector.
Next, we show that the extremal variogram from \cref{eq:intro_extremal_variogram},
the covariance matrices of extremal functions,
and the corresponding precision matrices are precisely the matrices in this diagram.
The same diagram also represents the compositional variation array from \cref{eq:intro_variation_array},
the centred log-ratio covariance matrix,
and the additive log-ratio covariance matrices of compositional data.
Building on this algebraic connection,
in \cref{sec:connections} we investigate several transfers of methods between the two fields.
The most direct transfer is a class of intrinsic logistic-normal graphical models for compositional data.
These models are obtained by imposing sparsity on the singular precision
matrix of centred log-ratios,
in analogy with \HR{} graphical models \citep{engelke2020,hentschel2023} in extreme value theory.
In other direction, we discuss how popular dimensionality reduction methods such as log-ratio analysis and weighted log-ratio analysis
\citep{aitchison_pca_1983,greenacreLewi2009distributional}
translate from compositional data analysis to extreme value theory.
\cref{sec:outlook} gives an outlook on further connections
and generalizations beyond the kernel spanned by the one-vector.
Technical details on oblique projections and \MP{} inverses
are provided in \cref{sec:obliqueProjectionsDetails,sec:MPInverse},
and mathematical proofs in \cref{sec:proofs}.

Certain parallels between extreme value theory
and compositional data analysis have been noted in the literature before
\citep{lhautRootzenSegers2025,colesTawn1991,hansonCarvalhoChen2017,kakampakouWadsworth2025},
but our work is the first to formalize this connection
by establishing
a shared algebraic framework for the covariance structures in both fields
that allows a systematic transfer of statistical methods.

%% file: tex/02_Background.tex
\subsection{Extreme value theory}

Modelling the multivariate extremes of a random vector can be achieved by several equivalent or closely related ways,
depending on whether one considers componentwise maxima,
threshold exceedances,
or
point process limits.
In this paper,
we focus on the threshold exceedance perspective,
where the limiting object is a multivariate generalized Pareto distribution.
Within this class of distributions,
the \HR{} distribution is a particularly tractable model,
since it is closely related to the multivariate normal distribution,
and allows for useful parametric representations of the dependence structure.

\subsubsection{\MultGP{} distributions}
\label{sec:backgroundMultGP}
We consider the setting of generalized Pareto distributions with exponential margins.
More precisely,
we assume that the one-dimensional margins
of the random vector $X$ follow standard exponential distributions,
and are interested in the distribution of the vector $Y$,
arising as the limit
\begin{equation}
    \label{eq:mpd:exponential}
    \P\lr{Y \leq z}
    :=
    \limit{u} \P\lr{
        X - u\onevec \leq z
        \mid
        X \not\leq u \onevec
    }
    =
    \frac{
        \expmeasure(\Rd \setminus \halfopen{-\inftyvec, z})
    }{
        \expmeasure(\Rd \setminus \halfopen{-\inftyvec, 0})
    }
    ,
\end{equation}
supported on the set $\setm{y \in \Rd}{y \not\leq \zerovec}$
and
characterized by the so-called exponent measure~$\expmeasure$ \citep{rootzen2006}.
The expression inside the limit is illustrated for a finite
threshold $u$ in the left two panels of \cref{fig:evtSamples}.
We furthermore assume that there is no mass at $-\infty$ in any component of $Y$,
implying that all components of $X$ are asymptotically dependent.
The exponent measure $\expmeasure$ and,
if it is absolutely continuous, the corresponding density $\expdens$,
satisfy the homogeneity properties
\begin{equation}
    \label{eq:homogeneityEVT}
    \expmeasure(A + \alpha \onevec)
    =
    \exp(-\alpha)
    \expmeasure(A)
    , \qquad
    \expdens(y + \alpha \onevec)
    =
    \exp(-\alpha)
    \expdens(y)
    ,
\end{equation}
for $\alpha \in \R$ and measurable sets $A$.
Due to these homogeneity properties,
the distribution of $Y$ can be decomposed into \citep{rootzen2018}
$Y = S + E \onevec$,
where $S$ is the angular part,
supported on $\setm{x \in \Rd}{\max(x) = 0}$,
and $E$ is the radial part,
following a standard exponential distribution independent of $S$.
In some contexts,
it is more insightful to work with the conditioned random vectors
$Y^v = Y \mid v\T Y > 0$,
for some vector $v \geq \zerovec$.
These random vectors can be decomposed similarly into
\begin{equation}
    \label{eq:vExtremalFunction}
    Y^v
    =
    W^v + E \onevec
    ,
\end{equation}
with $W^v$ now supported on the hyperplane $\ocompl{v} = \setm{x}{v\T x = 0}$,
and $E$ again standard exponential and independent of $W^v$ \citep{hentschel2025}.
In particular,
the canonical unit vectors $v = \ek$ yield useful representations,
since $W^{\ek}$ is supported on the hyperplane $\setm{x}{x_k = 0}$,
and can thus be identified with a random vector $W\slr{k}$ supported on $\R^{d-1}$,
obtained by removing the $k$-th component of $W^{\ek}$.

\begin{figure}
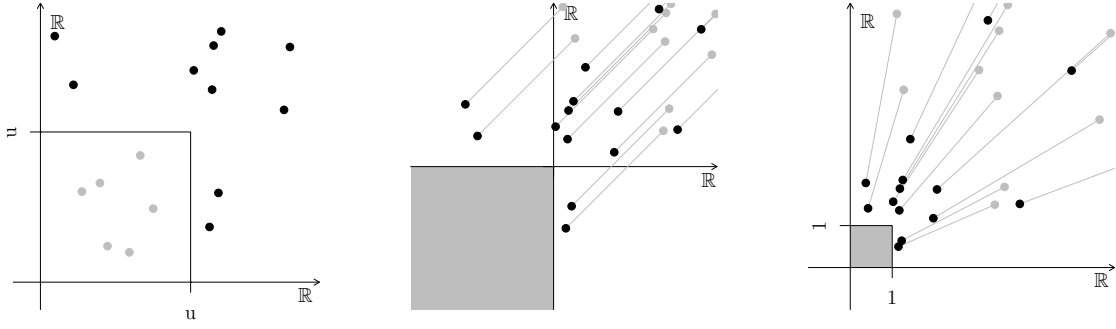

    \centering
    \resizebox{0.32\textwidth}{!}{%
        \inputtikzR{multivariateEVT/evtLogThresholdEx2}
    }
    \resizebox{0.32\textwidth}{!}{%
        \inputtikzR{multivariateEVT/evtLogThresholdEx3}
    }
    \resizebox{0.32\textwidth}{!}{%
        \inputtikzR{multivariateEVT/evtThresholdEx3}
    }
    \caption{%
        Illustration of the limit in \cref{eq:mpd:exponential}
        with a finite dataset and threshold $u$.
        The original distribution $X$ is thresholded (left),
        and shifted by $\onevec u$ (centre),
        resulting in a support of $\setm{y}{\max(z) \geq 0}$,
        for $u \to \infty$.
        On standard Pareto margins (cf. \cref{remark:ParetoMargins}),
        the distribution is rescaled instead of shifted (right),
        and the limiting support is
        $\setm{z \in \halfopen{0, \infty}^d}{\max(z) \geq 1}$.
    }
    \label{fig:evtSamples}
\end{figure}

\begin{figure}
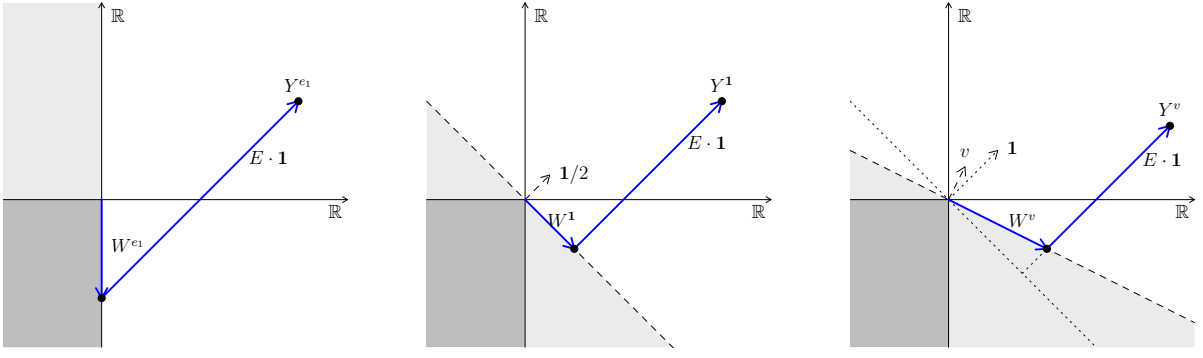

    \center
    \resizebox{0.30\textwidth}{!}{%
        \inputtikzR{densityProjections/e1.tex}
    }
    \hfill
    \resizebox{0.30\textwidth}{!}{%
        \inputtikzR{densityProjections/OneVec.tex}
    }
    \hfill
    \resizebox{0.30\textwidth}{!}{%
        \inputtikzR{densityProjections/V_onePerp.tex}
    }
    \caption{%
        Illustration of the decomposition
        of the conditioned vector $Y^v$
        into extremal function $W^v$
        and radial part $E$ for choices
        $v = \ek[1]$ (left),
        $v = \onevec$ (centre),
        and a general vector $v \in \Rd[2]$ (right).
        The support of $Y^v$ is the non-grayed out half-space.
    }
    \label{fig:extremalFunctions}
\end{figure}

There are various ways to describe and estimate the dependence structure of a \multGP{} distribution
\citep[e.g.,][]{colesHeffernanTawn1999,Resnick2004}.
One such way is through the extremal variogram \citep{engelkeVolgushev2022},
which is defined as the element-wise average over
the matrices $\Gamma^{\ek} \in \Rd$
with entries
\begin{equation}
    \label{eq:defVariogram}
    \Gamma_{ij}^{\ek}
    =
    \Var \lr{Y_i^{\ek} - Y_j^{\ek}}
    ,
\end{equation}
for $k = 1, \dots, d$.
In \cite{hentschel2025},
this definition is generalized to
general vectors $v \geq \zerovec$ and averaged
over arbitrary sets of such vectors.

\begin{remark}
    \label{remark:ParetoMargins}
    Often, the \multGP{} distribution is defined with standard Pareto margins
    instead of exponential margins.
    The transformation between the two is given by $Y = \log Z$,
    where $Z$ is the random vector with Pareto margins,
    and the additive homogeneity properties mentioned above
    are multiplicative in this case.
    \cref{fig:evtSamples}, right panel,
    illustrates the corresponding thresholding operation and support.
    We will use exponential margins in the following,
    but sometimes mention Pareto margins, as they allow for easier
    interpretation of asymptotically independent components,
    which correspond to zero values instead of $-\infty$ in the exponential case.
\end{remark}

\subsubsection{The \HR{} distribution}
\label{sec:backgroundHR}
A class of extreme value distributions
with particularly nice covariance structures
is the \HR{} distribution.
The max-stable version of this
distribution was introduced in \cite{hueslerReiss1989}
as a limit of $\max_{i=1, \dots, n} X_i$,
where $X_i \simiid \normal(0, \Sigma_n)$,
and $(1 - (\Sigma_n)_{ij}) \log n \longrightarrow \tfrac{1}{4} \Gamma_{ij}$.
The corresponding \multGP{} distribution is fully parametrized by the matrix $\Gamma$,
and its exponent measure is absolutely continuous
with density
\begin{equation*}
    \expdens(y; \Gamma)
    =
    \exp(-y_k)
    \cdot
    \varphi_{d-1}(\tilde{y}_{\without{k}}; \Sigma\slr{k})
    .
\end{equation*}
Here, $\varphi_{d-1}$ denotes the $(d-1)$-dimensional centred normal density
with covariance $\Sigma\slr{k}$ and
\begin{alignat}{6}
    \Sigma\slr{k}_{ij}
    &=
    \half(
        \Gamma_{ik} + \Gamma_{jk} - \Gamma_{ij}
    )
    , \quad &
    i,j &\neq k
    \label{eq:SigmaKFromGamma}
    , \\
    \tilde{y}_{i}
    &=
    y_i - y_k + \half \Gamma_{ik}
    , \quad &
    i &\neq k
    \nonumber
    .
\end{alignat}
Moreover, the distribution of the extremal functions $W\slr{k}$
is multivariate normal with covariance $\Sigma\slr{k}$.
Alternatively, this exponent measure can be defined through
a point process construction involving
copies of a Gaussian distribution with covariance $\Sigma$,
whose variogram is $\Gamma$,
i.e., $\Gamma_{ij} = \Sigma_{ii} + \Sigma_{jj} - 2 \Sigma_{ij}$
\citep{engelke2015}.

\citet{engelke2020} introduce graphical models for extremes,
and show that the \HR{} distribution allows the graphical structure of
a model to be read off from the precision matrices
$\Theta\slr{k} := \lr{\Sigma\slr{k}}\inv$,
using the relation
\begin{equation*}
    \excondind{Y_i}{Y_j}{Y_{\without{i,j}}}
    \Longleftrightarrow
    \Theta\slr{k}_{ij} = 0
    , \quad
    i, j \neq k
    .
\end{equation*}
For $i = k$ or $j = k$,
this zero-condition needs to be satisfied by the corresponding row sum of $\Theta\slr{k}$,
or verified with a different pivot $k$.

In \cite{hentschel2023},
a different parametrization of the \HR{} distribution is suggested
through the non-invertible precision matrix $\Theta\sonevec \in \Rdd$,
defined by
\begin{equation}
    \label{eq:ThetaFromThetaK}
    \Theta\sonevec_{ij}
    =
    \Theta\slr{k}_{ij}
    , \quad
    \text{~for~some~}
    k \neq i,j
    .
\end{equation}
Notably, the value of $\Theta\slr{k}_{ij}$ does not depend on the choice of $k$.
Furthermore,
they show that this matrix is also the \MP{} inverse of the non-invertible covariance matrix
\begin{equation}
    \label{eq:SigmaFromGamma}
    \Sigma\sonevec
    :=
    \lr{\Id - \onevec\onevec\T/d}
    \lr{-\half \Gamma}
    \lr{\Id - \onevec\onevec\T/d}
    .
\end{equation}
For general vectors $v \geq \zerovec$,
the $v$-extremal functions $W^v$ from \cref{eq:vExtremalFunction},
are supported on the hyperplane $\ocompl{v}$
and have covariance matrix
\begin{equation}
    \label{eq:SigmaVFromGamma}
    \Sigma^v
    =
    (\Id - \onevec v\T/\onevec\T v)
    \Sigma\sonevec
    (\Id - v \onevec\T/\onevec\T v)
    .
\end{equation}

\subsection{Compositional data analysis}
\label{sec:backgroundCoDA}

Compositional data analysis provides a framework for data vectors whose
absolute scale is not meaningful,
so that information is carried only by the relative sizes of the components.
One standard approach
to modelling compositional data
is to work in Aitchison geometry,
where operations and distances are defined in terms of log-ratios.
In the following, we give a brief introduction to this geometry,
see also \citet{aitchison_statistical_1986}
for a standard reference
and \citet{greenacre_compositional_2021} for a recent review.

\subsubsection{Aitchison geometry}
\label{sec:backgroundCoDAAitchison}

We consider samples of non-negative vectors $w \in [0, \infty)^d$
where only the relative proportions of the components matter,
not the absolute values themselves.
Since the transformations introduced below require the use of
logarithms and division,
we further require all components to be strictly positive,
that is $w \in (0, \infty)^d$;
a brief discussion of how to handle zeros is given in \cref{sec:zerosSmallValues}.
In the context of compositional data analysis,
the vector $w$ is called a $d$-part composition,
and the components $w_i$ are called the parts of the composition.
As these proportions are invariant under scaling,
the data is usually rescaled to the interior of the $(d-1)$-simplex,
denoted
$\simplex = \setm{x \in \Rd}{x_i > 0, \sum_i x_i = 1}$
via the closure operator $\cl$ defined as
\begin{equation}
    \cl: \lr{0, \infty}^d
    \rightarrow
    \simplex
    , \quad
    (w_1, \dots, w_d)
    \mapsto
    \lr{
        \frac{w_1}{\sum_i w_i},
        \dots,
        \frac{w_d}{\sum_i w_i}
    }
    .
\end{equation}
A two-dimensional example of this closure operation
is shown in the left two panels of \cref{fig:codaCLR_2}.
However,
this unit-sum constraint leads to several problems,
such as spurious correlations and dependence on the
overall choice of components included in the analysis
(i.e., which quantities are considered as parts of the overall sum $\sum_i w_i$).
These problems are resolved in \cite{aitchison_statistical_1986}
by considering the following log-ratio transformations.
For a compositional data vector $x \in \simplex$,
the pairwise log-ratio between $x_i$ and $x_j$ is defined as
\begin{equation*}
  \LR(i, j)
  =
  \log \lr{\frac{x_i}{x_j}}
  , \quad
  i,j \in \set{1, \dots, d}
  ,
\end{equation*}
and computing all distinct pairwise log-ratios
with respect to a common denominator
yields the additive log-ratio
\begin{equation*}
  \ALR{j}(x)
  =
  (\LR(i,j))_{i \neq j}
  \in
  \Rd[d-1]
  .
\end{equation*}
Replacing $x_j$ by the geometric mean of $x$
yields the centred log-ratio
$\CLR(x) \in \Rd$ with entries
\begin{equation*}
  \CLR(x)_i
  =
  \log\lr{\frac{
    x_i
  }{
    \lr{\prod_{j=1}^d x_j}^{1/d}
  }}
  , \quad
  i \in \set{1, \dots, d}
  .
\end{equation*}
The computation of these transformations is illustrated in \cref{fig:codaCLR_2}.
Choosing more general linear combinations for the numerator and denominator,
we obtain log-contrasts,
defined as $\sum_i a_i \log x_i$
for some coefficients $a_i$ with $\sum_i a_i = 0$.
These transformations have many desirable properties,
such as being invariant to scaling
and allowing the definition of a meaningful Hilbert space structure
on the simplex~$\simplex$.

\begin{figure}
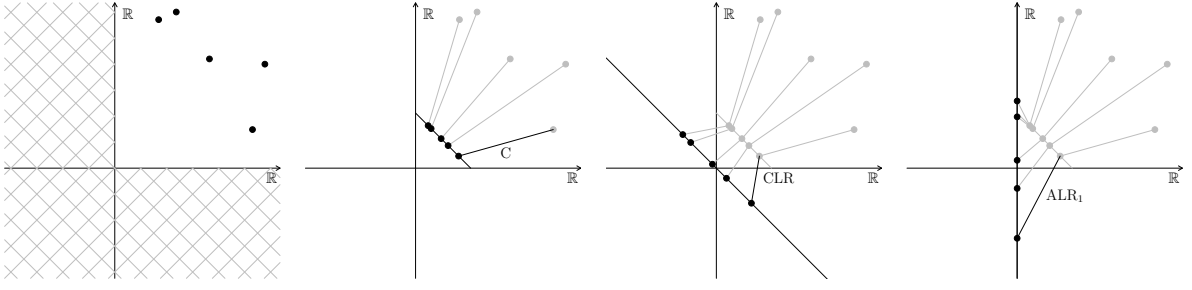

    \centering
    \resizebox{0.24\textwidth}{!}{%
        \inputtikzR{coda/x_raw}
    }
    \resizebox{0.24\textwidth}{!}{%
        \inputtikzR{coda/x_closed}
    }
    \resizebox{0.24\textwidth}{!}{%
        \inputtikzR{coda/CLR_1}
    }
    \resizebox{0.24\textwidth}{!}{%
        \inputtikzR{coda/ALR1_1}
    }
    \caption{%
        Original sample of positive values $w \in (0, \infty)^2$ (left),
        compositionally closed vector $x = \cl(w) \in \Delta$ (centre left),
        followed by computation of the centred log-ratios
        $\CLR(x)$ (centre right),
        and the additive log-ratios $\ALR{1}(x)$,
        embedded in $\Rd$ by inserting a zero in the first component
        (right).
    }
    \label{fig:codaCLR_2}
\end{figure}

\subsubsection{Covariance structures}
\label{sec:backgroundCoDACovStructures}

In \citet[][Chapters~4-5]{aitchison_statistical_1986},
various parametrizations of covariance structures
for compositional data are discussed.
Below, we use notation consistent with other parts of this work;
a comparison to the notation used by \citet{aitchison_statistical_1986}
is provided in \cref{sec:notationComparison}.
For these definitions,
we assume that $x$ is a random vector
supported on the simplex $\simplex$.
The compositional variation array
is defined as
the matrix containing the variance of all pairwise log-ratios,
that is,
the matrix $\Gamma \in \Rdd$ with entries
\begin{equation}
    \label{eq:compositionalVariation}
    \Gamma_{ij}
    =
    \Var \lr{\log \frac{x_i}{x_j}}
    , \quad
    i,j \in \set{1, \dots, d}
    .
\end{equation}
A common scalar summary of the variance of a composition is the
total variability, defined as
\begin{equation}
    \label{eq:totalVariability}
    \textrm{TotVar}
    =
    \sum_{i < j} \Gamma_{ij}
    =
    \frac{1}{2} \sum_{i,j} \Gamma_{ij}
    .
\end{equation}
Sometimes, a multiplicative factor of $1/d$ is included in this definition
\citep{greenacre_variable_selection_2018}.
The log-ratio covariance matrix with pivot $k$
is defined as the covariance matrix
of the additive log-ratio transformed data $\ALR{k}(x)$,
that is,
the matrix $\Sigma\slr{k} \in \R^{(d-1)\times(d-1)}$
with entries
\begin{equation}
    \label{eq:LogratioCovariance}
    \Sigma\slr{k}_{ij}
    =
    \Cov\lr{
        \log \frac{x_i}{x_k},
        \log \frac{x_j}{x_k}
    }
    , \quad
    i,j \in \set{1, \dots, d}\without{k}
    .
\end{equation}
Analogously,
the centred log-ratio covariance matrix $\Sigma\sonevec \in \Rdd$ is defined
as the covariance matrix of the centred log-ratio transformed data $\CLR(x)$,
with entries
\begin{equation}
    \label{eq:CenteredLogratioCovariance}
    \Sigma\sonevec_{ij}
    =
    \Cov\lr{
        \CLR(i),
        \CLR(j)
    }
    , \quad
    i,j \in \set{1, \dots, d}
    .
\end{equation}
To describe the dependence structure of the original data before closure,
the covariance structure of a $d$-part basis $w$
(interpreted as a random vector on $(0, \infty)^d$)
is defined as the covariance matrix of the logarithm of $w$,
that is,
$\Sigma \in \Rdd$ with entries
\begin{equation*}
    \Sigma_{ij}
    =
    \Cov\lr{
        \log w_i,
        \log w_j
    }
    .
\end{equation*}

%% file: tex/03_ObliqueProjections.tex
As illustrated above,
in both extreme value theory and compositional data analysis,
there are many different ways to represent the dependence structure
of the underlying distribution or dataset.
All of these representations are rank-deficient in some sense,
either by depending only on differences and ratios,
such as $\Gamma$ in \cref{eq:defVariogram,eq:compositionalVariation},
having one fewer dimension than the data itself,
such as $\Sigma\slr{k}$ in \cref{eq:SigmaKFromGamma,eq:LogratioCovariance},
or by having non-trivial kernels,
such as $\Sigma\sonevec$ in \cref{eq:SigmaFromGamma,eq:CenteredLogratioCovariance}.
Furthermore,
it is not immediately clear how these different representations relate to each other,
as, for example, the transformation between $\Theta\slr{k}$ and $\Theta\sonevec$
in \cref{eq:ThetaFromThetaK} is not an elementary mathematical operation at first glance,
and computations of log-ratios for compositional data
involve multiple steps of closure operations, log transformations, and ratio calculations.

In this section,
we define a general framework
that allows us to formulate these different representations of the dependence structure
using a small set of mathematical operations,
most importantly oblique projections and pseudoinverses.
In \cref{sec:obliqueProjections},
we define and relate these operations in a very general setting,
allowing for arbitrary kernels and images of the projections.
Then, in \cref{sec:covarianceStructures}, we show how the covariance structures
from extreme value theory and compositional data analysis
fit into this framework, using a common kernel of $\onevec\R$.

\subsection{Oblique projections of vectors}
\label{sec:projectionsOfVectors}
We consider the space $\Rd$ of $d$-dimensional real vectors.
Unless specified otherwise,
vectors, denoted by lowercase letters, are from $\Rd$,
matrices, denoted by uppercase letters, are from $\Rdd$,
and subspaces, also denoted by uppercase letters, usually $U$, $V$, $W$,
are subspaces of $\Rd$.
Two subspaces $U, V$ are complementary if they intersect only in the zero vector
and their direct sum is the whole space $\Rd$.
\begin{definition}
    \label{def:obliqueProjection}
    Let $U$, $V$ be complementary subspaces.
    Then $\PUV$ denotes the oblique projector along $U$ onto $V$,
    that is, the unique matrix satisfying
    \begin{equation*}
        \PUV u
        =
        \zerovec
        \quad \forall
        u \in U
        , \qquad
        \PUV v
        =
        v
        \quad \forall
        v \in V
        .
    \end{equation*}
\end{definition}
The kernel $U$ is also called the along-space,
and the image $V$ the onto-space of $\PUV$.
See \cref{fig:illustration_lemma_pvwpvw},
left and right panel,
for an illustration of the effect of an oblique projection on a vector.
It is well known that oblique projections
are idempotent and that
any idempotent matrix is an oblique projection
along its kernel onto its image
\citep[e.g.,][]{Afriat1957}.

In the following,
we identify projections
$\PXYo{U}{V}$
by their kernel, say $U$,
and the orthogonal complement $V$ of their image $V\orth$.
With this notation, $U$ and $V$ have the same dimension,
and expressions that relate different projections become more symmetric,
for instance the transpose of a projection
$\lr{\PXYo{U}{V}}\T = \PXYo{V}{U}$,
see also \cref{lemma:projectorMPInv}.
In particular, when the kernel is one-dimensional,
and the image is a hyperplane,
this convention is very natural,
and we simplify notation
by denoting a subspace by the vector that spans it,
e.g., $\Pxyo{u}{v} := \PXY{(u\R)}{(\set{v}\orth)}$.

Multiplying projections with different kernels and images
does not yield a projection in general \citep[e.g.,][]{takaneYanai1999}.
However, if the projections share a common kernel or image,
the product is again a projection.
This is a rather basic property of projections,
but it is fundamental for the relationships
between different covariance representations in later sections,
which is why we state it as a \namecref{lemma:pvwpvw} below.
\begin{lemma}
    \provenlabel{lemma:pvwpvw}
    Let $\ocompl{V}, \ocompl{W}$ each be complementary to $U$.
    Then
    \begin{align*}
        \PUVo
        \PUWo
        &=
        \PUVo
        , \\
        \lr{\PUVo\restrictto{W\orth}}\inv
        &=
        \PUWo\restrictto{V\orth}
        .
    \end{align*}
\end{lemma}
In other words, a sequence of projections with common kernel
reduces to the last (i.e, left-most) projection,
and, shown for example by considering transposes,
a sequence of projections with common image
reduces to the first (i.e., right-most) projection.
Since $\PUWo$ is the identity on $W\orth$,
this implies the bijectivity relationship in the last equation.
See \cref{fig:illustration_lemma_pvwpvw} for a two-dimensional illustration
of these relationships.

\begin{figure}
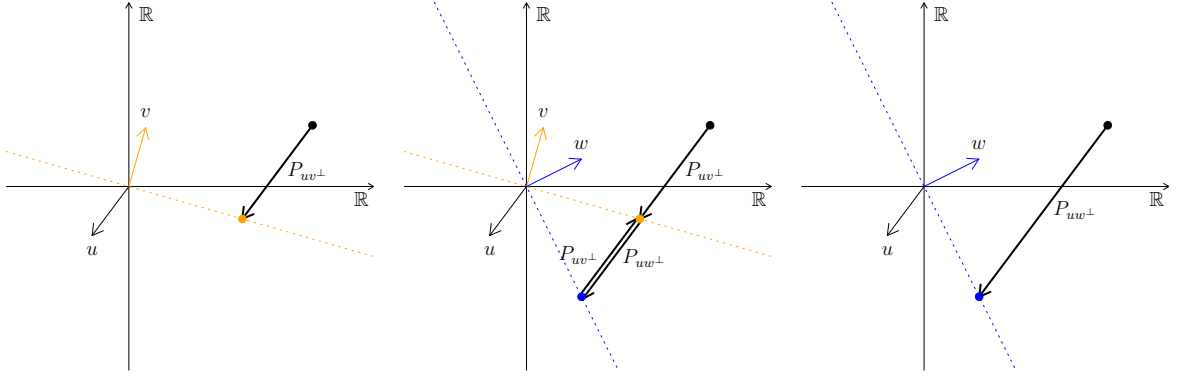

    \centering
    \resizebox{0.32\textwidth}{!}{%
        \inputtikzR{uvwProjections/proj_uv}
    }
    \resizebox{0.32\textwidth}{!}{%
        \inputtikzR{uvwProjections/proj_uvwinv}
    }
    \resizebox{0.32\textwidth}{!}{%
        \inputtikzR{uvwProjections/proj_uw}
    }
    \caption{%
        Illustration of \cref{lemma:pvwpvw}.
        The left and right panels show the effect of $\Puvo$ and $\Puwo$
        on a general vector $x \in \Rd[2]$.
        The middle panel illustrates that subsequent projections with common kernel
        reduce to the last projection,
        yielding a bijective relationship between~$\ocompl{v}$ and~$\ocompl{w}$.
    }
    \label{fig:illustration_lemma_pvwpvw}
\end{figure}
\subsection{Oblique projections of positive semidefinite matrices}

When applying a projection $\PUVo$ to a centred random vector $Y$,
the covariance of the projected vector is related to
the original covariance matrix operation by
\begin{equation}
    \label{eq:covarianceOfProjection}
    \E\lr{\PUVo Y \lr{\PUVo Y}\T}
    =
    \E\lr{\PUVo Y Y\T \PVUo}
    ,
\end{equation}
motivating the following definition.
\begin{definition}
    \label{def:pivw}
    For matrix $A \in \Rdd$
    and complementary subspaces $U, \ocompl{V}$,
    define the ``covariance projection''
    $\pUVo$ as
    \begin{equation*}
        \pUVo(A)
        =
        \PUVo A \PVUo
        \in
        \Rdd
        .
    \end{equation*}
\end{definition}
Note that this operation is defined for general matrices $A$,
but we are mostly interested in the symmetric, positive semidefinite case,
occurring as covariance matrices of random vectors.
To denote the sets that these matrices belong to,
let $\setpsdany$ be the set of all symmetric positive semidefinite matrices,
and define
\begin{align}
    \setpsd{V}
    &=
    \setm{
        A \in \setpsdany
    }{
        \kernel(A) = V
    }
    \label{eq:setpsdV}
    , \\
    \setpsdt{V}
    &=
    \setm{
        A \in \setpsdany
    }{
        \image(A) + V = \Rd
    }
    \label{eq:setpsdtV}
    .
\end{align}
By construction,
the kernel of $\pUVo(A)$ necessarily contains $V$,
the set $\setpsd{V}$ therefore consists of all ``maximal rank'' covariance
matrices after this operation.
Furthermore,
the following lemma shows that
the set $\setpsdt{U}$ consists of positive semidefinite
matrices that map to such a matrix under $\pUVo$.
In particular,
the set $\setpsdt{U}$ contains all strictly positive definite matrices
and is a superset of
$\setpsd{V}$ for any $\ocompl{V}$ complementary to $U$.
When working with Gaussian distributions,
we are also interested in precision matrices, defined as the (pseudo)inverse of the covariance matrix.
To understand the relationship between different covariance projections and pseudoinverses,
we give the following result.
\begin{lemma}
    \provenlabel{lemma:propertiesSymmetricProjection}
    Let $\ocompl{V}, \ocompl{W}$ be complementary to $U$.
    Then for $A \in \setpsdt{U}$,
    we have $\pUVo(A) \in \setpsd{V}$
    and
    \begin{alignat}{6}
        \pUVo
        \pUWo
        &=
        \pUVo
        \label{eq:symmetricPvwpvw}
        , \\
        \lr{\pUVo\restrictto{\setpsd{W}}}\inv
        &=
        \pUWo\restrictto{\setpsd{V}}
        \label{eq:symmetricProjectionRestrictedInverse}
        , \\
        \pUVo(A)\pinv
        &=
        \pVVo(A\pinv)
        \in
        \setpsd{V}
        ,
        \qquad&
        \text{for }
        A \in \setpsd{U}
        \label{eq:projectionMPInvPSD1}
        , \\
        \pUUo(A)\pinv
        &=
        \pVUo(A\pinv)
        \in
        \setpsd{U}
        ,
        \qquad&
        \text{for }
        A \in \setpsd{V}
        .
        \label{eq:projectionMPInvPSD2}
    \end{alignat}
\end{lemma}
The first two relationships
in \cref{lemma:propertiesSymmetricProjection}
are a direct consequence of \cref{lemma:pvwpvw},
whereas the relationships involving the pseudoinverse
are less immediate and require verification of the
Moore-Penrose conditions from \cref{def:MPInverse}.
These relationships are represented by the commutative
diagram in \cref{fig:commutesUVpsd},
and for the choice $U=\onevec\R$ in
the bottom two rows of the diagram in \cref{fig:commutesVOnevecGamma}.

%% file: tex/04_CovarianceStructures.tex
\subsection{The variogram}
\label{sec:variogram}

The maps in the previous section also define equivalence classes of covariance structures,
in the sense that a matrix $A \in \setpsd{U}$
can be represented by the matrix $\pUVo(A) \in \setpsd{V}$,
without any loss of information, and vice versa.
For arbitrary kernels $U$
these equivalence classes do not have a clear interpretation,
but for specific choices they can be meaningful.
A very important such example with $U = \onevec\R$
arises from the variogram corresponding to a covariance matrix,
defined by the mapping
\begin{alignat*}{6}
    \myFunctionDefinitionInline{\vario}{\Rdd}{\Rdd}{S}{
        \onevec \diag\lr{S}\T
        + \diag\lr{S} \onevec\T
        - 2 S
    }
    .
\end{alignat*}
This map can be defined for all matrices in $\Rdd$,
but is mostly relevant for the set of symmetric positive (semi-)definite matrices $\setpsdany$.
If a variogram matrix has full rank,
it is (strictly) \cnd{} and belongs to the set $\setcnd$ defined as
\begin{align*}
    \setcnd
    &=
    \setm{
        M \in \Rdd
    }{
        M = M\T,
        \diag\lr{M} = \zerovec,
        v\T M v < 0
        \,\forall\,
        \zerovec \neq v \perp \onevec
    }
    .
\end{align*}
The relationships between different covariance and precision matrices corresponding
to the same variogram are described by the following result.
As above, we simplify notation for one-dimensional subspaces
and denote them by the vector that spans them,
e.g., $\setpsd{v} := \setpsd{(v\R)}$.
\begin{proposition}
    \provenlabel{prop:variogramProperties}
    Let $v$ be any vector such that $v \not\perp \onevec$.
    The variogram map $\vario$ has the following properties.
    \begin{enumerate}
        \item \label{prop:variogramProperties1}
            For $A \in \setpsdt{\onevec}$, we have $\vario(A) \in \setcnd$,
            and for $A \in \setpsdany \setminus \setpsdt{\onevec}$, we have $\vario(A) \notin \setcnd$.
        \item \label{prop:variogramProperties2}
            The map $\vario$ is invariant under covariance projections along $\onevec\R$
            onto $\ocompl{v}$, i.e.,
            \begin{align*}
                \vario\lr{\pevo(A)}
                &=
                \vario(A)
                , \quad
                \forall A \in \Rdd
                .
            \end{align*}
        \item \label{prop:variogramProperties3}
            Restricted to domain $\setpsd{v}$ and codomain $\setcnd$,
            the map $\vario$ is bijective with inverse $-\half\pevo$.
            Moreover, $\pevo(-\half \vario(A)) = \pevo(A)$ for all $A \in \Rdd$.
        \item \label{prop:variogramProperties4}
            For $\Sigma\sonevec \in \setpsd{\onevec}$ and $\Sigma^v \in \setpsd{v}$,
            satisfying $\vario(\Sigma\sonevec) = \vario(\Sigma^v)$,
            the precision matrices $\Theta\sonevec = (\Sigma\sonevec)\pinv$ and $\Theta^v = (\Sigma^v)\pinv$ satisfy
            \begin{align*}
                \Theta\sonevec
                &=
                \pveo(\Theta^v)
                , \qquad
                \Theta^v
                =
                \pvvo(\Theta\sonevec)
                .
            \end{align*}
    \end{enumerate}
\end{proposition}

This \namecref{prop:variogramProperties} fully describes the covariance structures
corresponding to a \cnd{} variogram.
Any covariance matrix in $\setpsdt{\onevec}$ corresponds to a \cnd{} variogram,
whereas any other covariance matrix does not map to a (strictly) \cnd{} variogram.
Since the map between $\setpsd{v}$ and $\setcnd$ is bijective,
there is exactly one covariance matrix in each set $\setpsd{v}$
corresponding to a given \cnd{} variogram.
Furthermore,
the different covariance matrices corresponding to the same variogram
can be obtained from each other by covariance projections along $\onevec\R$,
or from the variogram by applying $-\half \pevo$.
Finally, the precision matrices corresponding to the same variogram
are related by covariance projections as well,
though with different kernels than the covariance matrices.
These relationships are summarized in the diagram in \cref{fig:commutesVOnevecGamma}.
In the following, we show that these matrices and relationships
occur naturally in both extreme value theory and compositional data analysis.

\begin{figure}
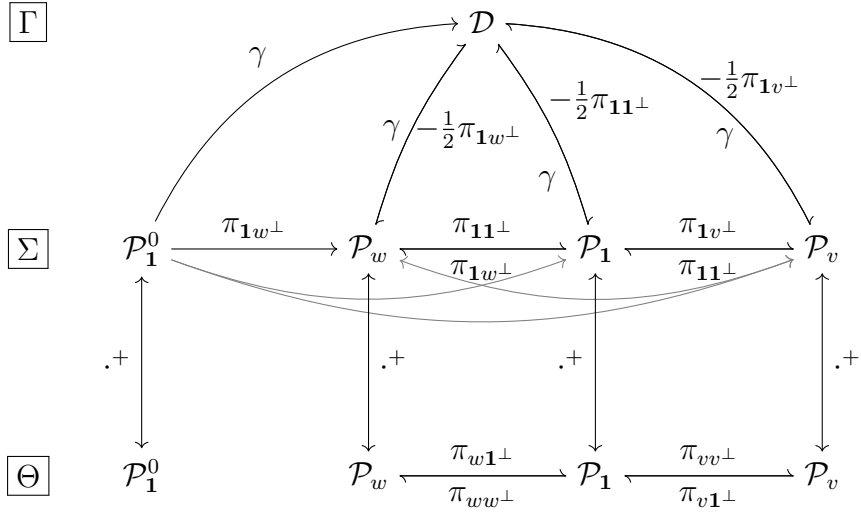

    \centering
    \inputtikz{commutesVOnevecGammaWithW}
    \caption{%
        Commutative diagram showing the relationship between covariance matrices,
        precision matrices, and variograms.
        The top node is the variogram representation $\Gamma \in \setcnd$,
        the second row contains covariance matrices $\Sigma^v \in \setpsd{v}$
        for different choices of $v$,
        and the third row contains the corresponding precision matrices $\Theta^v \in \setpsd{v}\T$.
        Labelled arrows indicate a mapping from the set at its tail to the set at its head.
        In the second row, unlabelled arrows with target $\setpsd{u}$ for some $u$
        are projections $\pexo{u}$.
        Double-headed arrows indicate bijections,
        with $\cdot\pinv$ and $\cdot\inv$ being self-inverses,
        otherwise each direction is labelled on the side of the corresponding arrowhead.
        Omitted arrows indicate maps that do not simplify nicely.
        Commutativity implies that any two directed paths
        with the same start and end point define the same map.
        For example, starting from $\setpsd{v}$ in the second row,
        one can either first project to $\setpsd{\onevec}$
        and then apply the pseudoinverse,
        or first apply the pseudoinverse and then project as indicated in the precision row,
        yielding the identity
        $(\peeo(\Sigma^v))\pinv = \pveo((\Sigma^v)\pinv) \in \setpsd{\onevec}$
        for $\Sigma^v \in \setpsd{v}$.
    }
    \label{fig:commutesVOnevecGamma}
\end{figure}

\subsection{Extreme value theory}
\label{sec:covarianceStructuresEVT}

When considering \multGP{} distributions on exponential margins,
the vector $\onevec$ plays a special role,
both in the defining limit equation \cref{eq:mpd:exponential}
and the resulting homogeneity properties of the exponent measure and density,
as shown in \cref{eq:homogeneityEVT}.
Shifting the data by a multiple of $\onevec$ only changes the shared
radial component of the distribution, but not the dependence structure.
This is reflected also in the covariance structures of the \HR{} distribution,
which are all related by covariance projections along $\onevec\R$,
as the following result shows.

\begin{corollary}
    \provenlabel{cor:covarianceStructuresEVT}
    Let $Y$ be a random variable with \multGP{} \HR{} distribution,
    parametrized by $\Gamma \in \setcnd$,
    and denote $\Sigma^v = \pevo(-\half \Gamma) \in \setpsd{v}$.
    Then this $\Sigma^v$
    is the covariance of the $v$-extremal function of $Y$
    as defined in \cref{eq:SigmaVFromGamma}.
    For the matrix $\Sigma\sonevec$ from \cref{eq:SigmaFromGamma}
    we have
    \begin{equation*}
        \Sigma\sonevec
        =
        \peeo(-\half \Gamma)
        ,
    \end{equation*}
    and the matrices $\Sigma\slr{k}$ from \cref{eq:SigmaKFromGamma} can
    be obtained by dropping the $k$-th row and column
    from the matrices
    \begin{equation*}
        \Sigma^{\ek}
        =
        \pxyo{\onevec}{\ek}(-\half \Gamma)
        =
        \pxyo{\onevec}{\ek}(\Sigma\sonevec)
        .
    \end{equation*}
\end{corollary}

Thus, the \HR{} variogram and covariance matrices
are precisely the matrices in \cref{prop:variogramProperties},
satisfying the same relationships.
In particular,
their precision matrices can also be related by covariance projections as
\begin{equation*}
    (\Sigma\sonevec)\pinv
    =
    \Theta\sonevec
    =
    \pxeo{\ek}(\Theta^{\ek})
    , \qquad
    (\Sigma^{\ek})\pinv
    =
    \Theta^{\ek}
    =
    \pxxo{\ek}(\Theta\sonevec)
    .
\end{equation*}

\begin{remark}
    \label{remark:inverseDroppingZeros}
    Since the kernel of the matrices
    $\Sigma^{\ek}, \Theta^{\ek} \in \Rdd$
    is spanned by $\ek$,
    their $k$-th row and column are necessarily filled with zero entries.
    Dropping these rows and columns yields invertible matrices
    $\Sigma\slr{k}, \Theta\slr{k} \in \Rdd[(d-1)]$.
    Taking the pseudoinverse of a block diagonal matrix
    is equivalent to taking the (pseudo)inverse of each block separately,
    and the pseudoinverse of a zero block is again a zero block.
    Therefore, the pseudoinverse relationship $(\Sigma^{\ek})\pinv = \Theta^{\ek}$
    is equivalent to $(\Sigma\slr{k})\inv = \Theta\slr{k}$.
\end{remark}

The \namecref{cor:covarianceStructuresEVT} shows that the
variogram, covariance, and precision matrices
related to \HR{} distributions
and their extremal functions
can be expressed through only three mappings,
the variogram map $\vario$,
the \MP{} inverse,
and the covariance projection $\puvo$
for different choices of $u$ and $v$.
In particular the relationship between different precision matrices
as projections has not been described before, to the best of our knowledge.

\begin{remark}
    \label{remark:generalExtremalFunctions}
    It is important to note that for general \multGP{} distributions,
    the covariance matrices of the extremal functions are
    not exactly related by the mappings described above.
    Since they are defined as projections of different random variables,
    specifically
    $W^v = \Pevo Y^v$ for $Y^v \deq Y \mid v\T Y > 0$,
    they are related by
    an exponential tilting of the distribution \citep{hentschel2025}.
    The \HR{} distribution is a special case with Gaussian extremal functions,
    whose covariance structure is invariant under this tilting.
\end{remark}

One way to resolve this issue
would be to consider projections $\Pevo Y$
of the original random variable $Y$,
or equivalently of the spectral random vector $S$ from
\cref{sec:backgroundMultGP}.
However, using $\tilde W^v \deq \Pevo Y$, the distribution of $Y$ cannot be expressed any more
as an independent sum $\tilde W^v + E\onevec$,
as the support of $E$ depends on the realization of $\tilde W^v$.

\subsection{Compositional data analysis}
\label{sec:covarianceStructuresCoDA}

In compositional data analysis,
only proportions of the components are analysed,
meaning that the data is invariant to multiplicative scaling.
When considering logarithmic transformations of the data,
this invariance translates to invariance to shifts by a multiple of $\onevec$,
again suggesting that projections with kernel $\onevec\R$ might be relevant.

In fact, when working with log-transformed data $y = \log w$,
the closure operator $\cl$ can be rewritten as a shift along $\onevec\R$ as
$\lcl(y) := \log(\cl(\exp y)) = y - \onevec \log \sum_i \exp y_i$.
The following \namecref{lemma:logratiosAsProjections}
shows that the log-ratio transformations from \cref{sec:backgroundCoDAAitchison}
can similarly be expressed as oblique projections.

\begin{lemma}
    \provenlabel{lemma:logratiosAsProjections}
    The additive and centred log-ratio transformations can be expressed as
    \begin{equation*}
        \ALR{k}:
        w
        \mapsto
        (\Pexo{\ek}(\log w))_{i \neq k}
        , \qquad
        \CLR:
        w
        \mapsto
        \Peeo(\log w)
        .
    \end{equation*}
\end{lemma}

Since the closure operator is equivalent to a shift along
$\onevec\R$ on the log-transformed data,
these projections can be applied directly to $\log w$,
without changing the result.
\cref{fig:codaCLR_3} shows the computation of the $\CLR$ and $\ALR{1}$ transformations
as projections of the log-transformed data,
for the same example as in \cref{fig:codaCLR_2}.
\begin{figure}
    \centering
    \resizebox{0.24\textwidth}{!}{%
        \inputtikzR{coda/x_raw}
    }
    \resizebox{0.24\textwidth}{!}{%
        \inputtikzR{coda/x_log}
    }
    \resizebox{0.24\textwidth}{!}{%
        \inputtikzR{coda/CLR_2}
    }
    \resizebox{0.24\textwidth}{!}{%
        \inputtikzR{coda/ALR1_2}
    }
    \caption{%
        Same example as in \cref{fig:codaCLR_2},
        but represented as taking logarithms,
        followed by the projection $\Pxyo{\onevec}{\onevec}$,
        and the projection $\Pxyo{\onevec}{\kev{1}}$.
    }
    \label{fig:codaCLR_3}
\end{figure}
Recalling the definitions of the compositional variation array $\Gamma$
and the related covariance matrices from
\cref{sec:backgroundCoDACovStructures}
we note that these can similarly be related through
covariance projections and the variogram map,
as shown in the following result.

\begin{corollary}
    \provenlabel{cor:covarianceStructuresCoDA}
    Let $\Sigma \in \setpsdt{\onevec}$ be the covariance of the log-transformed data $\log w$.
    Then the compositional variation array $\Gamma$ from \cref{eq:compositionalVariation}
    can be expressed as
    $\Gamma = \vario(\Sigma)$,
    the centred log-ratio covariance matrix $\Sigma\sonevec$ from \cref{eq:CenteredLogratioCovariance}
    can be expressed as
    \begin{align*}
        \Sigma\sonevec
        &=
        \peeo(-\half \Gamma)
        =
        \peeo(\Sigma)
        ,
    \end{align*}
    and the additive log-ratio covariance matrices $\Sigma\slr{k}$ from \cref{eq:LogratioCovariance} can
    be obtained by dropping the $k$-th row and column
    from the matrices
    \begin{align*}
        \Sigma^{\ek}
        =
        \pexo{\ek}(-\half \Gamma)
        =
        \pexo{\ek}(\Sigma)
        =
        \pexo{\ek}(\Sigma\sonevec)
        .
    \end{align*}
\end{corollary}

Similar to the \HR{} case described above,
this implies that all relationships stated in \cref{prop:variogramProperties}
also hold for the compositional variation array and related covariance matrices.
The precision matrices corresponding to these covariance matrices
are less relevant in the descriptive analysis of compositional data,
but come up when modelling compositional data for example
with logistic-normal distributions \citep[][Chapter~6]{aitchison_statistical_1986},
and also satisfy the relationships described in the \namecref{prop:variogramProperties}.

%% file: tex/05_Methods.tex
\input{examples/danube_numbers.tex}
\input{examples/gemas_numbers.tex}
\input{examples/logistic_gauss_numbers.tex}

In the previous section 
we showed that the covariance structures in compositional data analysis
and extreme value theory
are related by the same family of mappings between covariance matrices and variograms.
So far, this connection is only an algebraic one,
since the data-generating processes and stochastic models
underlying the two settings are different.
However, having established the algebraic connection,
we can investigate on a case-by-case basis
whether certain models and methods
from one setting can be meaningfully transferred to the other setting.

The general idea behind these transfers is to
identify quantities (random variables, covariance matrices, etc.) in the two settings
that correspond to each other
in the sense that they take the same place in the general
algebraic structure from \cref{prop:variogramProperties,fig:commutesVOnevecGamma}.
Many of the matching quantities are already represented by the same
symbols in previous sections,
but for clarity
we give some explicit examples of corresponding quantities below.
\begin{align*}
    \text{Extremal variogram~}
    \Gamma\sonevec
    &\cong
    \text{Compositional variation array~}
    \Gamma
    , \\
    \text{Extremal function~}
    W\slr{k}
    &\cong
    \text{ALR-transformed data~}
    \ALR{k}(X)
    , \\
    \text{Extremal function~}
    W\sonevec
    &\cong
    \text{CLR-transformed data~}
    \CLR(X)
    , \\
    \text{Covariance of extremal function~}
    \Sigma\slr{k}
    &\cong
    \text{ALR covariance matrix~}
    \Sigma\slr{k}
    , \\
    \text{Covariance of extremal function~}
    \Sigma\sonevec
    &\cong
    \text{CLR covariance matrix~}
    \Sigma\sonevec
    .
\end{align*}
As discussed in \cref{remark:generalExtremalFunctions},
for extreme value distributions that are not \HR{},
quantities such as the extremal variogram $\Gamma^v$
depend on the choice of $v$.
We follow the approach of \citet{wan2025}
and use the symmetric case $v = \onevec$ as ``canonical'' choice for this vector;
an alternative would be to consider averages over different choices of $v$,
such as in the original definition of the extremal variogram
by \citet{engelkeVolgushev2022}.

Using these correspondences,
we can for example investigate whether a method
that estimates the compositional variation array $\Gamma$ based on CLR-transformed data,
can be adapted to estimate the extremal variogram $\Gamma\sonevec$ based on the extremal function $W\sonevec$.
Following this approach, we provide examples for method transfers in both directions.
In \cref{sec:intrinsicGraphicalModels},
we introduce a family of intrinsic graphical models for compositional data,
which is based on \HR{} graphical models in extreme value theory,
and show how to estimate the graphical structure of such models,
using \cglearn{}, an adaptation of the \eglearn{} algorithm from \citet{engelke2022a}.
In \cref{sec:pca},
we explore different dimensionality reduction methods
that were developed for compositional data,
and show how they can be applied to extreme value theory.

\subsection{Intrinsic logistic-normal graphical models}
\label{sec:intrinsicGraphicalModels}

\subsubsection{The model}
Sparse models for covariance structures are a common approach
in statistics,
as they allow for interpretable models and
more efficient inference compared to dense models \citep{engelke2021a}.
In compositional data analysis,
there are several different approaches to defining and estimating
such sparse covariance structures.
One approach,
used in methods such as \sparcc{} \citep{friedmannAlm2012_sparcc}
and \cclasso{} \citep{fangEtAl2015_cclasso},
is to assume and estimate a sparse covariance matrix
$\Sigma$ for the unobserved
log-transformed population.
Another approach is to
assume a normal distribution with
sparse precision matrix $\Theta$ for the log-transformed population.
The zeros in this precision matrix imply conditional independence
relations between the components of the log-transformed population,
commonly represented as graphical models.
This approach is taken for the methods
\spieceasi{} \citep{KurtzEtAl2015}
and
\gcoda{} \citep{fangHuangZhaoDeng2017},
which use different assumptions and estimation procedures to enforce sparsity in $\Theta$.

In the following, we suggest a different approach,
assuming that the
centred log-ratio transformed data
follows a (degenerate) normal distribution with sparse precision matrix $\Theta\sonevec$.
This is an assumption on the distribution of the observed
compositional data itself, rather than on an unobserved population.
Due to the degeneracy of the distribution,
the sparsity pattern of $\Theta\sonevec$ does not directly imply
conditional independence relations as in standard Gaussian graphical models.
Instead, such a model corresponds to an intrinsic Gaussian graphical model
in the sense of \citet{rue2005},
and we show in \cref{lemma:logisticNormalALR} that ALR-transformed data for each pivot
follows a non-degenerate Gaussian graphical model with the same sparsity pattern.

\begin{definition}
    Let $\Theta\sonevec \in \setpsd{\onevec}$ be a (sparse) precision matrix,
    and $G=(V,E)$ an undirected graph on $\set{1,\dots,d}$,
    satisfying $\Theta\sonevec_{ij} = 0$ for all $\set{i,j} \notin E$.
    We call $X$ an intrinsic logistic-normal graphical model
    with graph $G$, if,
    for some $\mu \in \ocompl{\onevec}$,
    it has the stochastic representation
    \begin{equation*}
        X
        =
        \cl\lr{\exp(Y_{\CLR})}
        , \qquad\text{where}\qquad
        Y_{\CLR}
        \sim
        \normal(\mu, (\Theta\sonevec)\pinv)
        .
    \end{equation*}
\end{definition}

One possible model for the original parts $W$ is given
by a log-normal distribution with covariance matrix
$\Sigma = (\Theta\sonevec)\pinv + s \onevec \onevec\T$,
where $s$ is a variance parameter for the shared scale factor.
Note that this model generally has a dense covariance
matrix $\Sigma$ and precision matrix
$\Theta = \Theta\sonevec + \onevec \onevec\T / (d^2 s)$.
Using \cref{lemma:logratiosAsProjections},
\cref{cor:covarianceStructuresCoDA}, and \cref{prop:variogramProperties},
we obtain the following result about the distribution of the ALR-transformed data.

\begin{lemma}
    \provenlabel{lemma:logisticNormalALR}
    Let $X$ be an intrinsic logistic-normal graphical model with graph $G$
    and precision matrix $\Theta\sonevec$.
    Then the ALR-transformed data $Y\slr{k} = \ALR{k}(X)$ is distributed as
    \begin{align*}
        Y\slr{k}
        &\sim
        \normal(\mu\slr{k}, \Sigma\slr{k})
        ,
    \end{align*}
    where $\mu\slr{k} = (\Pexo{\ek}\mu)_{i \neq k}$ and
    $\Sigma\slr{k}$ is obtained by removing the $k$-th row and column from
    $\Sigma^{\ek} = \pexo{\ek}((\Theta\sonevec)\pinv)$.
    Furthermore,
    $Y\slr{k}$ is a Gaussian graphical model whose graph is
    the subgraph of $G$ obtained by removing node $k$ and its incident edges.
\end{lemma}

\subsubsection{Inference}

A natural approach to estimate the sparse precision matrix $\Theta\sonevec$
is to use an $\ell_1$-penalized likelihood approach,
maximizing
\begin{align*}
    f(\Theta\sonevec)
    &=
    \log\pdet{\Theta\sonevec}
    + \half\trace(\Theta\sonevec \hat\Gamma)
    - \lambda \sum_{i\neq j} \abs{\Theta\sonevec_{ij}}
    ,
\end{align*}
where $\hat\Gamma$ is the empirical variogram of the data.
Note that
$-\half\trace(\Theta\sonevec\hat\Gamma) = \trace(\Theta\sonevec\Sigma\sonevec)$.
However, it has been shown that this approach does not
learn a sparse graph \citep{yingEtAl2023},
and a number of alternative approaches have been suggested in
extremal graphical modelling literature.

One such algorithm is \eglearn{} \citep{engelke2022a}
which estimates the sparsity pattern of~$\Theta\sonevec$
by applying a graphical base learner $\Acal$ to each
extremal function $Y\slr{k}$ separately,
and then aggregating the results by a majority vote across
all $k \neq i,j$ for each possible edge $\set{i,j}$.
Since the extremal functions have full-rank precision matrices
$\Theta\slr{k} \in \Rdd[(d-1)]$,
known structure estimation algorithms
such as neighbourhood selection \citep{MeinshausenBuhlmann2006}
and graphical lasso \citep{YuanLin2007}
can be used as base learner $\Acal$.

In \cref{alg:cglearn}, we propose \cglearn{},
an adaptation of the \eglearn{} algorithm to the setting of compositional data analysis.
The idea is the same as in extreme value theory,
but instead of computing the extremal functions,
we compute the ALR-transformed data $\ALR{k}(X)$,
to which we apply the graphical base learner $\Acal$.
As long as the base learner eventually identifies the correct
sparsity pattern of each $\Theta\slr{k}$,
\cref{lemma:logisticNormalALR} ensures that
the correct sparsity pattern of~$\Theta\sonevec$ is identified as well.

\begin{algorithm}[H]
    \caption{\cglearn{} (adaptation of \eglearn{})}
    \label{alg:cglearn}
    \begin{algorithmic}[1]
        \Statex \textbf{Input:} \parbox[t]{0.85\linewidth}{%
            Compositional data $X\in(0,\infty)^{n\times d}$,\\
            Adjacency matrix base learner
            $\mathcal{A}: \Rd[n \times (d-1)] \rightarrow \set{0,1}^{(d-1) \times (d-1)}$
        }
        \Statex \textbf{Output:} Estimated graph $\widehat{G}=(V,\widehat{E})$

        \State $\widehat{E}\leftarrow\emptyset$

        \For{$k\in V$}
            \State $\widehat{A}^{(k)}\leftarrow\mathcal{A}(\ALR{k}(X))$
            \State $\widetilde{A}^{(k)}\leftarrow$ augment $\widehat{A}^{(k)}$ with a zero row/column at position $k$
        \EndFor

        \For{$i,j\in V,\;i<j$}
            \State $C\leftarrow\#\{k\in V\without{i,j}:\;\widetilde{A}^{(k)}_{ij}=1\}$
            \If{$C>(d-2)/2$}
                \State $\widehat{E}\leftarrow\widehat{E}\cup\set{\set{i,j}}$
            \EndIf
        \EndFor

        \State \Return $\widehat{G}=(V,\widehat{E})$
    \end{algorithmic}
\end{algorithm}

Importantly, this algorithm only estimates the graphical structure of the model,
that is, the sparsity pattern of $\Theta\sonevec$,
not the actual values of the non-zero entries.
An approach for subsequent parameter estimation for a model with this sparsity pattern
based on graphical completions of partial variogram matrices is suggested in \cite{hentschel2023}.
Other approaches to sparse \HR{} graphical models in extreme value theory
are introduced by \cite{wan2023graphical}, \cite{roettger2023}, and \cite{lederer2023extremes}.

\subsubsection{Simulation and application}
In \cref{fig:cglearnLogisticGauss},
we illustrate the application of \cglearn{} to a simulated dataset.
The data consists of $\logisticgaussValueNDataRows$ samples
from a $\logisticgaussValueNNodes$-dimensional
intrinsic logistic-normal graphical model,
parametrized by a randomly generated sparse precision matrix $\Theta\sonevec$
with $\logisticgaussValueNEdgesBase$~non-zero off-diagonal entries (edges).
As base learner $\Acal$, we use neighbourhood selection
\citep{MeinshausenBuhlmann2006}.
The number of edges depends on a penalty parameter;
choosing it such that the number of edges is the same as in the true graph,
\cglearn{} identifies $\logisticgaussValueNEdgesEglearnBaseAgree$
out of the $\logisticgaussValueNEdgesBase$ edges correctly.

\begin{figure}
    \centering
    \resizebox{0.3\textwidth}{!}{
        \inputtikzR{logistic_gauss/graph_base}
    }
    \includegraphics[width=0.3\textwidth]{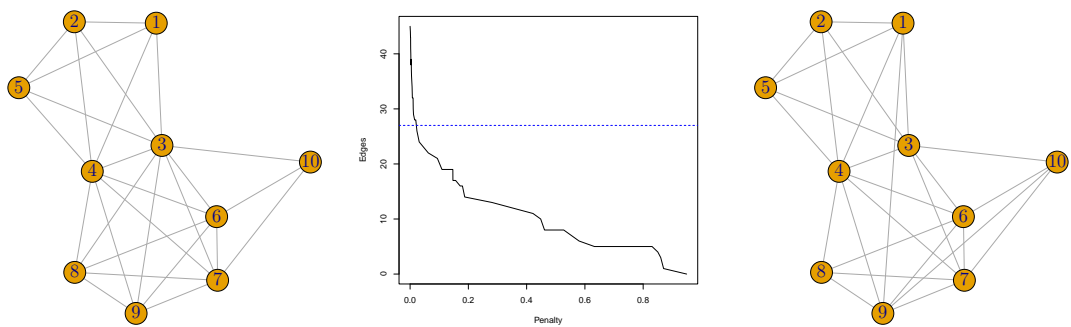}
    \resizebox{0.3\textwidth}{!}{
        \inputtikzR{logistic_gauss/graph_eglearn}
    }
    \caption{%
        True graph structure of the simulated intrinsic logistic-normal graphical model (left),
        and
        number of edges in the graph estimated by \cglearn{} for this dataset,
        as a function of the penalty parameter (centre).
        The true edge count, $\logisticgaussValueNEdgesBase$,
        is marked by a horizontal line;
        the corresponding graph estimated by \cglearn{} is shown in the right panel.
    }
    \label{fig:cglearnLogisticGauss}
\end{figure}

\cref{fig:cglearnGemas} illustrates the application of \cglearn{}
to the \gemas{} dataset,
a compositional dataset of geochemical measurements
available in the R package \robCompositions{} \citep{robCompositions}.
For illustrative purposes,
we also apply the method \cclasso{} \citep{fangEtAl2015_cclasso} to the same dataset,
using a grid of penalty parameters for both methods,
in order to compare graphs with equal numbers of edges.
Since there is no ground truth for the underlying graph,
we compare the two methods for an ad-hoc choice of penalty parameters
yielding graphs with $\gemasValueNEdgesEglearn$ edges each.
For this choice,
the two methods agree on $\gemasValueNEdgesEglearnCClassoAgree$ edges.
The ratio of shared edges between the two methods as a function of
total number of edges in the estimated graphs is shown in the right panel of \cref{fig:cglearnGemas}.
Note that the two methods are based on different underlying models and assumptions,
and do not aim to estimate the same graph.
Similar comparisons with $\gcoda$ \citep{fangHuangZhaoDeng2017}
and on different datasets give comparable results.

\begin{figure}
    \centering
    \resizebox{0.65\textwidth}{!}{
        \inputtikzR{gemas/graph_eglearn}
        \hspace{0.05\textwidth}
        \inputtikzR{gemas/graph_cclasso}
    }
    \includegraphics[width=0.3\textwidth]{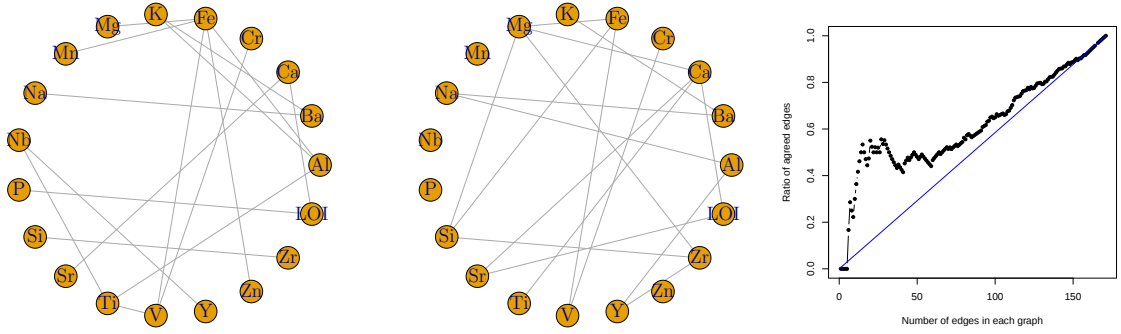}
    \caption{%
        Graph structures estimated for the \gemas{} dataset
        by \cglearn{} (left) and \cclasso{} (centre).
        Vertices correspond to components of the composition,
        and edges indicate non-zero entries in the estimated precision
        and correlation matrix, respectively.
        The right panel shows
        the ratio of shared edges between \cglearn{} and \cclasso{}
        when choosing penalty parameters such that
        both graphs have the number of edges indicated on the X-axis of the plot.
    }
    \label{fig:cglearnGemas}
\end{figure}

\subsection{Dimensionality reduction for extremes}
\label{sec:pca}
In this \namecref{sec:pca},
we explore different dimensionality reduction methods
that were developed for compositional data,
and how they can be applied to extreme value theory.
The methods discussed here,
are based mostly on the compositional variation array $\Gamma$,
and the centred log-ratio covariance matrix $\Sigma\sonevec$.

\subsubsection{Log-ratio analysis for extremes}
Principal component analysis (PCA)
is widely used for dimensionality reduction and visualization.
Naive approaches in compositional data analysis apply standard PCA to the closed data,
which is problematic
since it ignores the unit-sum constraint and relative nature of compositional data
\citep{aitchison_pca_1983}.
Working with additive log-ratio transformed data is also problematic
due to the dependence on the choice of the pivot component.
One symmetric solution,
known as log-ratio analysis (LRA)
is derived by considering linear combinations of all log-ratios
$\log\lr{x_i/x_j}$,
which is equivalent to using so called log-contrasts,
that is,
convex combinations of the form $\sum_i v_i \log x_i$,
with $\sum_i v_i = 0$.
The variance maximizing log-contrasts are then given by
the eigenvectors of the CLR covariance matrix
$\Sigma\sonevec \in \setpsd{\onevec}$ \citep{aitchison_pca_1983,aitchison_statistical_1986}.

In extreme value theory,
\citet{wan2025} introduces an approach for PCA
based on the covariance matrix of the so-called profile random vectors $U$.
In the notation of this paper,
the vector $U$ corresponds to the extremal function $W\sonevec$,
and its covariance matrix is $\Sigma\sonevec$.
The PCA based approximation of $W\sonevec$
is therefore given by
\begin{equation*}
    W\sonevec
    \approx
    \PX{q_{1}} W\sonevec
    +
    \PX{q_{2}} W\sonevec
    +
    \dots
    ,
\end{equation*}
where $q_1, q_2, \dots$ are the eigenvectors of $\Sigma\sonevec$,
and $\PX{s}$ denotes the orthogonal projection onto the span of $s$.
Hence, on an algebraic level,
this approach is equivalent to LRA in compositional data analysis,
as both are based on an eigendecomposition of the same covariance matrix $\Sigma\sonevec$.

Based on this connection,
we can transfer advanced methods from LRA
to the setting of extremes.
One such method is weighted LRA \citep{greenacreLewi2009distributional,hronEtAl2021},
which assigns different weights to different parts of a composition,
as well as each observation in the dataset.
In \cref{fig:lraDanube},
we illustrate the application of this method to the Danube dataset,
which is a commonly used reference dataset in extreme value theory \citep{asadi2015extremes}.
To do so, we use the R package \easycoda{} \citep{greenacre_book_2018},
which provides an implementation of LRA,
as well as plotting functions used to create biplots of the results.
In a biplot,
observations are represented by labelled points
(or just the labels)
and components by labelled arrows,
such that the scalar product between an arrow and a label
gives the approximate value of the corresponding centred log-ratio for that observation
\citep{gabriel1971}.
In the weighted biplot,
we put more emphasis on high-flow parts of the river network
by assigning the average absolute water flow at each station
as weights for the columns.
Both versions recover the flow structure of the river network,
with connected stations, such as $\set{V25, V26, V27}$, being close to each other in the biplot.
The weighted version puts more emphasis on the main flow of the river,
resulting in a clearer separation of the largest station ($V1$) from others.

\begin{figure}
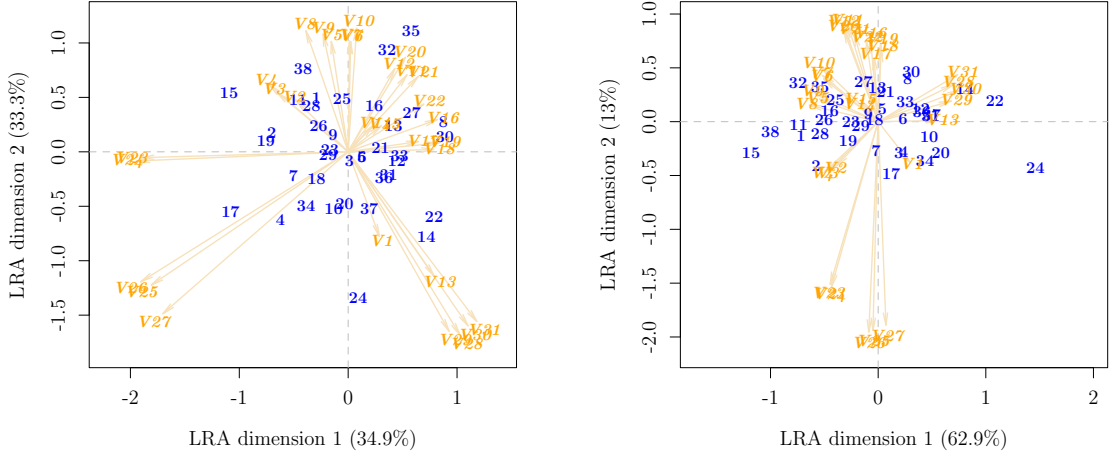

    \centering
    \resizebox{0.48\textwidth}{!}{
        \inputtikzR{lra/lra_unweighted}
    }
    \resizebox{0.48\textwidth}{!}{
        \inputtikzR{lra/lra_weighted}
    }
    \caption{%
        Biplots of the first two log-ratio principal components for the Danube dataset,
        without weights (left) and with weights
        based on absolute flow volumes per station (right).
        The flow graph illustrating the connections between
        stations is shown in \cref{fig:dimReductionDanube}.
        In each plot, blue labels correspond to observations in the dataset,
        and labelled arrows correspond to the components.
        The geometric arrangement is such that the scalar product between an arrow and a label
        gives the approximate value of the corresponding centred log-ratio for that observation.
        The part of the total variance explained by each principal component
        is shown in parentheses in the axis labels.
    }
    \label{fig:lraDanube}
\end{figure}

\subsubsection{Interpretable dimension reduction}
\label{sec:interpretableDimReduction}

Several other dimension-reduction methods in compositional data analysis are based on
the total log-ratio variability of a composition,
defined in \cref{eq:totalVariability}
as the sum of all log-ratio variances $\Gamma_{ij}$, $i<j$.
One such method, proposed by \citet{greenacre_variable_selection_2018},
is a stepwise log-ratio selection procedure.
At each step, the method selects a pair of components $i,j$ such that the log-ratio
$\log(x_i/x_j)$ explains the largest amount of remaining variation
in a redundancy-analysis criterion.

The selected log-ratios can be represented as edges in a graph,
often as a tree when $d-1$ independent log-ratios are chosen,
as shown in the right panel of \cref{fig:dimReductionDanube}.
This graph should be interpreted as a visualization of the
result of a variable selection procedure,
rather than a conditional-independence graph.
In particular, selected edges tend to correspond to component pairs
with high log-ratio variance,
whereas edges in conditional-independence graphs
are typically associated with smaller variances between them.

A complementary approach is the amalgamation-based hierarchical clustering method
of \citet{greenacre2019}.
Instead of selecting highly informative log-ratios,
this method successively merges pairs of components into amalgamations,
such that each merge causes the smallest possible loss of total log-ratio variability.
The result is a hierarchical clustering of the components,
with early amalgamations corresponding to parts that can be combined
with little loss of relative information.

Both methods are descriptive and do not require a stochastic model
for the data-generating process.
They can therefore be transferred naturally to exploratory analysis
and dimension reduction in extreme value theory.
In this setting, pairwise log-ratios $\log(x_i/x_j)$
are replaced by pairwise differences of extremal functions,
for instance $W_i\sonevec - W_j\sonevec$,
and the total variability is computed by the same formula
for the corresponding extremal variogram
$\Gamma\sonevec$.
In \cref{fig:dimReductionDanube},
we illustrate both procedures on the Danube dataset,
again using implementations from the R package \easycoda{}
\citep{greenacre_book_2018}.

\begin{figure}
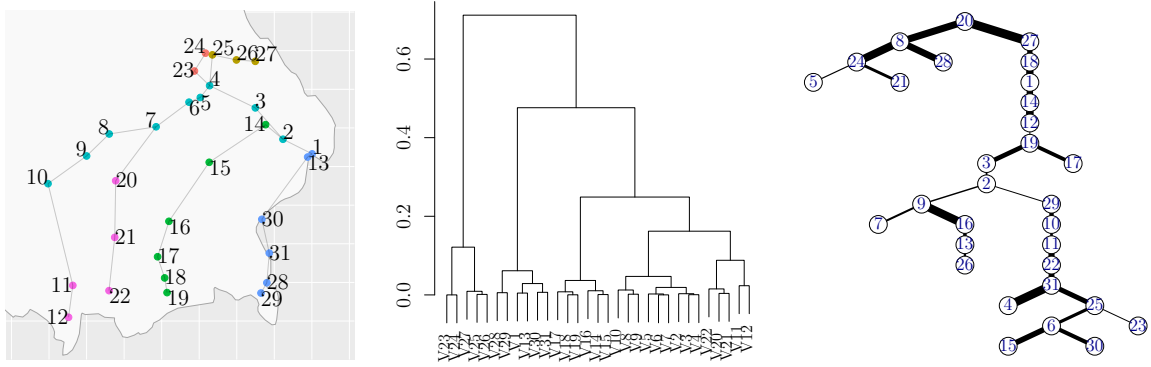

    \centering
    \resizebox{0.32\textwidth}{!}{%
        \inputtikzR{lr_selection/danube_aclust_clusters}
    }
    \resizebox{0.32\textwidth}{!}{%
        \inputtikzR{lr_selection/danube_aclust_dendrogram}
    }
    \resizebox{0.32\textwidth}{!}{%
        \inputtikzR{lr_selection/danube_step_graph}
    }
    \caption{%
        Application of amalgamation-based hierarchical clustering
        and stepwise log-ratio selection to the Danube dataset.
        The Y-axis in the dendrogram (centre) corresponds to
        the loss of total log-ratio variability
        caused by each amalgamation.
        In the true flow graph between the stations (left),
        the nodes are coloured according to the clusters obtained
        by the hierarchical clustering method,
        using an ad-hoc height of $0.1$ to define clusters.
        In the log-ratio selection graph (right),
        thicker edges indicate log-ratios selected earlier in the stepwise procedure.
    }
    \label{fig:dimReductionDanube}
\end{figure}

\subsubsection{Weighted variogram estimator}
\label{sec:weightedRows}

We close this section with a related estimation idea suggested by the row-weighted geometry of weighted LRA.
In practical applications of extreme value theory,
one usually considers only observations above some threshold $p \in (0,1)$ as extreme,
and uses these for inference.
This hard threshold is somewhat arbitrary,
and several approaches have been suggested to alleviate the issue
\citep{smithTawnColes1997,ledfordTawn1997,dupuis1998,stein2023weighted}.

In weighted LRA
\citep{greenacreLewi2009distributional,hronEtAl2021},
different weights are assigned to different observations in the dataset,
emphasizing observations more or less, based on some heuristic criterion.
Below, we transfer this approach to extreme value theory,
and define a weighted extremal variogram estimator.
Choosing the weighting scheme based on the overall magnitude of the observation,
the hard threshold described above can be softened
by giving more extreme observations higher weights.
For a data matrix $Y \in \Rd[n \times d]$,
sampled from (an approximation of) $Y\sonevec = Y | \set{\onevec\T Y > 0}$,
and unit-sum vector $r \in \Rd[n]$ of non-negative weights,
let
\begin{equation*}
    \widehat\Gamma^{r}_{ij}
    =
    \frac{1}{
        1 - \sum_{k=1}^n r_k^2
    }
    \sum_{k=1}^n
    r_k
    \lr{
        \lr{Y_{ki}-Y_{kj}}
        -
        \bar Y^{ij}_r
    }^2
    , \quad\text{with}\quad
    \bar Y^{ij}_r
    =
    \sum_{k=1}^n
    r_k
    \lr{Y_{ki}-Y_{kj}}
    .
\end{equation*}
The multiplicative factor $1/(1 - \sum_{k=1}^n r_k^2)$
is chosen such that the estimator is unbiased in the limiting model;
using uniform weights $r_k = 1/n$ recovers the
usual factor $n/(n-1)$ for the unbiased sample variance.
Considering the decomposition $Y\sonevec = W\sonevec + \onevec E$,
with $W\sonevec$ supported on $\ocompl{\onevec}$ and independent of $E$,
a natural choice for these weights is to be based on the value of $d\inv\onevec\T Y\sonevec$.
Indeed, in the limiting model,
this quantity is exactly $d\inv\onevec\T \onevec E = E$,
and thus independent of the extremal function~$W\sonevec$,
whereas in the practical setting of finite samples,
it can be used as a proxy for the ``extremeness'' of an observation.
\cref{fig:weightedVariogram} illustrates the performance of this weighted variogram estimator
for different thresholds and weighting schemes.
To compare the performance of different estimators,
we simulate $500$ datasets of size $n=1000$
from a $10$-dimensional max-stable distribution in the domain of attraction of a \HR{} \multGP{}
distribution with known variogram $\Gamma$.
To obtain approximate samples of $Y\sonevec$,
we threshold and margin-transform the simulated data
before conditioning on the event $\onevec\T Y > 0$.
MSE, bias, and variance of each estimator are then estimated from the $500$ simulated datasets
and shown in the figure.
We use equal weights of $r_k = 1/n$ as baseline,
comparing them to
weights proportional to the radial part $r_k \propto d\inv \onevec\T Y_{k\cdot}$,
and the cumulative distribution function
of the radial part $r_k \propto 1 - \exp(-d\inv \onevec\T Y_{k\cdot})$.
While we observe the usual bias-variance tradeoff with respect to the choice of threshold $p$,
the weighted estimators slightly outperform the unweighted estimator in terms of mean squared error
in this setting.
\begin{figure}
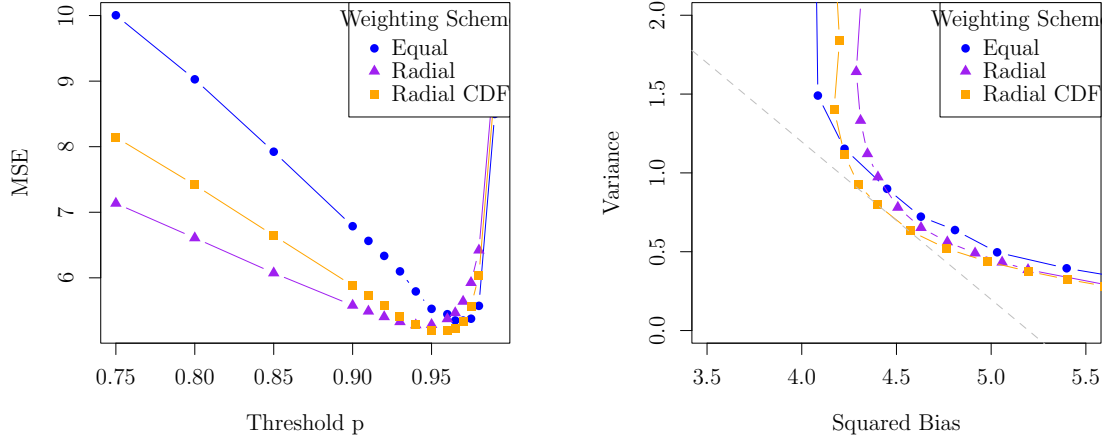

    \centering
    \resizebox{0.48\textwidth}{!}{
        \inputtikzR{weightedVariogram/mse_vs_p}
    }
    \resizebox{0.48\textwidth}{!}{
        \inputtikzR{weightedVariogram/bias_variance}
    }
    \caption{%
        Performance of the weighted variogram estimator with different weighting schemes
        as a function of the threshold $p$ (left),
        and corresponding bias-variance decompositions (right).
        We compare the unweighted estimator with equal weights $r_k = 1/n$ to two weighted estimators,
        one with weights proportional to the radial part
        $r_k \propto d\inv \onevec\T Y_{k\cdot}$,
        and one with weights proportional to the cumulative distribution function (CDF)
        of the radial part
        $r_k \propto 1 - \exp(-d\inv \onevec\T Y_{k\cdot})$.
    }
    \label{fig:weightedVariogram}
\end{figure}

%% file: examples/danube_numbers.tex
\newcommand{\danubeValueNNodes}{31}
\newcommand{\danubeValueNDataRows}{117}
\newcommand{\danubeValueNEdgesEglearn}{29}
\newcommand{\danubeValueNEdgesCClasso}{29}
\newcommand{\danubeValueNEdgesGCoda}{29}
\newcommand{\danubeValueNEdgesBase}{30}
\newcommand{\danubeValueNEdgesEglearnBaseAgree}{24}
\newcommand{\danubeValueNEdgesEglearnCClassoAgree}{19}
\newcommand{\danubeValuePctEdgesEglearnBaseAgree}{80.0}

%% file: examples/gemas_numbers.tex
\newcommand{\gemasValueNNodes}{19}
\newcommand{\gemasValueNDataRows}{2108}
\newcommand{\gemasValueNEdgesEglearn}{17}
\newcommand{\gemasValueNEdgesCClasso}{17}
\newcommand{\gemasValueNEdgesGCoda}{17}
\newcommand{\gemasValueNEdgesBase}{18}
\newcommand{\gemasValueNEdgesEglearnBaseAgree}{1}
\newcommand{\gemasValueNEdgesEglearnCClassoAgree}{8}
\newcommand{\gemasValuePctEdgesEglearnBaseAgree}{5.0}

%% file: examples/logistic_gauss_numbers.tex
\newcommand{\logisticgaussValueNNodes}{10}
\newcommand{\logisticgaussValueNDataRows}{1000}
\newcommand{\logisticgaussValueNEdgesEglearn}{27}
\newcommand{\logisticgaussValueNEdgesCClasso}{27}
\newcommand{\logisticgaussValueNEdgesGCoda}{27}
\newcommand{\logisticgaussValueNEdgesBase}{27}
\newcommand{\logisticgaussValueNEdgesEglearnBaseAgree}{25}
\newcommand{\logisticgaussValueNEdgesEglearnCClassoAgree}{18}
\newcommand{\logisticgaussValuePctEdgesEglearnBaseAgree}{92.0}

%% file: tex/06_Outlook.tex
\subsection{The problem of zeros/small values}
\label{sec:zerosSmallValues}

Throughout the paper,
we assumed strictly positive data in compositional data analysis,
to ensure that the log-ratio transformations are well-defined.
Similarly, in extreme value theory,
we assumed full asymptotic dependence,
which implies that the support of the \multGP{} distribution
on standard Pareto scale (discussed in \cref{remark:ParetoMargins})
is also supported on the strictly positive orthant.
Both fields have developed methods for handling zeros and small values,
occurring in practice,
such as zero-replacement strategies
and mixture models in compositional data analysis \citep{lubbeFilzmoserTempl2021,greenacre_compositional_2021},
and methods for identifying subfaces of the simplex with mass
in extreme value theory \citep{meyerWintenberger2024,chiapinoSabourinSegers2019,mourahibKirilioukSegers2024}.
Transferring such methods from one field to the other 
is a promising future research direction.

\subsection{Other kernels}
\label{sec:otherKernels}

The results in \cref{sec:obliqueProjections}
are given for general kernels $U$,
but the covariance structures in \cref{sec:covarianceStructures}
only involve the kernel spanned by the one-vector $\onevec$,
which is a natural choice in both
extreme value theory and compositional data analysis.

An interesting question is whether there are other natural directions to consider,
and whether these correspond to meaningful definitions of variogram matrices.
For example, \citet{egozcue_isometric_2004} suggests a ``linear compositional process'',
which includes an exponential decay rate $a_i$ for each component $i$,
shifting the log-transformed data by $a_i t$ for some time parameter $t$.
For unknown $t$ and non-trivial $a$, this model is identifiable only up to shifts along
the vectors $\onevec$ and $a = (a_1, \dots, a_d)\T$,
resulting in covariance structures with two-dimensional kernel $U = \laspan\lr{\onevec, a}$.

In the context of intrinsic graphical models,
the assumption of a one-dimensional kernel spanned by $\onevec$
always implies a connected underlying graph.
The precision matrix of a graph with two connected components,
say $C_1$ and $C_2$,
has the form
\begin{equation*}
    \Theta
    =
    \begin{pmatrix}
        \Theta_{C_1} & 0 \\
        0 & \Theta_{C_2}
    \end{pmatrix}
    ,
\end{equation*}
and necessarily has the indicator vectors $\indic{C_1}$ and $\indic{C_2}$ in its kernel.
Algebraically,
the corresponding covariance structures can still be described using
the results in \cref{sec:obliqueProjections} with this kernel.
However,
stochastically, this would correspond to a disconnected extremal graph \citep{eng_iva_kir}, and the interaction between the two components
has to be modelled separately,
for instance by a mixture model.

%% file: tex/A_Projections.tex
In the following we give more details on the properties
of oblique projections and the ``covariance projections'' defined in \cref{sec:obliqueProjections}.
In particular, we illustrate the relationships with commutative diagrams
and explicitly give intermediate results for projections of non-symmetric, indefinite matrices,
which are used to prove the relationships for symmetric, positive semidefinite matrices.
The proofs themselves are deferred to \cref{sec:proofsProjections},
as they are rather technical and not very insightful on their own.

First, we give an explicit formula to compute the projection matrix $\PUVo$
defined in \cref{def:obliqueProjection}
as well as expressions for its transpose and pseudoinverse.
\begin{lemma}
    \provenlabel{lemma:projectionFormula}
    Let $U$ and $\ocompl{V}$ be complementary subspaces,
    and let $A$, $B$ be matrices whose columns
    form bases of $\ocompl{V}$ and $\ocompl{U}$, respectively.
    Then, we have
    \begin{equation*}
        \PUVo
        =
        A (B\T A)\inv B\T
        .
    \end{equation*}
    Setting $B = A$ gives the formula for the orthogonal projection $\PXY{V}{\ocompl{V}}$.
    In the one-dimensional case where $U = u\R$ and $V = \set{v}\orth$,
    we have
    \begin{equation*}
        \Pxyo{u}{v}
        =
        \Id - \frac{uv\T}{v\T u}
        .
    \end{equation*}
\end{lemma}
\begin{lemma}
    \provenlabel{lemma:projectorMPInv}
    We have
    \begin{align}
        \PUVo\T
        &=
        \PVUo
        \nonumber
        , \\
        \PUVo\pinv
        &=
        \PX{\ocompl{V}}
        \PX{\ocompl{U}}
        \label{eq:projectionMPInv}
        , \\
        \PUVo\pinvT
        &=
        \PX{U}
        \PX{V}
        \nonumber
        .
    \end{align}
\end{lemma}

In order to illustrate \cref{lemma:pvwpvw},
we provide the following commutative diagram.
The illustrated relationships are rather simple,
but since the pattern of how the projections interact shows up
repeatedly in the following, we include it to make this part of later diagrams
more intuitive.
\begin{figure}[H]
    \centering
    \providebool{useXdiagrams}
    \setbool{useXdiagrams}{true}
    \inputtikz{commutesVectors}
    \caption{%
        Illustration of \cref{lemma:pvwpvw}.
        Arrows indicate a mapping from the set at its tail to the set at its head.
        Labels are not labelled individually,
        since, for a generic subspace $S$,
        all maps onto $\ocompl{S}$ are projections $\PXYo{U}{S}$.
        Commutativity implies that any two directed paths
        with the same start and end point define the same map.
        For example, comparing the direct path from $\Rd$ to $\ocompl{V}$
        with the path from $\Rd$ to $\ocompl{W}$ and then to $\ocompl{V}$,
        yields the relationship $\PUVo \PUWo = \PUVo$.
        Comparing the path from $\ocompl{W}$ to $\ocompl{V}$ and back
        to the empty path,
        implies that $\PUWo$ is the inverse of $\PUVo$ when restricting
        the domain and codomain accordingly.
    }
    \label{fig:commutesVectors}
\end{figure}

As pointed out in \cref{eq:covarianceOfProjection},
the covariance projection $\pUVo$
from \cref{def:pivw}
of a covariance matrix $A = \E(YY\T)$
can be expressed as the covariance of the projected random vector $\PUVo Y$.
Another, purely algebraic, construction can be obtained by
considering a matrix $X$ such that $XX\T = A$.
We then have
\begin{equation}
    \label{eq:identityGramProjection}
    \pUVo(A)
    =
    \PUVo XX\T \PVUo
    =
    (\PUVo X)(\PUVo X)\T
    .
\end{equation}
Similarly, the pseudoinverse of $A$ can be expressed in terms of
the transposed pseudoinverse of $X$,
using the identity
\begin{equation}
    \label{eq:identityGramPseudoinverse}
    (XX\T)\pinv
    =
    X\pinvT X\pinv
\end{equation}
Note that for general matrices $A$ and $B$, $(AB)\pinv \neq B\pinv A\pinv$.
We therefore proceed by establishing the relationships between projections
and pseudoinverses for matrices $X$ with suitable images.
We define the relevant sets of matrices,
corresponding to $\setpsd{U}$ and $\setpsdt{U}$ in \cref{sec:obliqueProjections},
as follows.

For a subspace $V$ and $k \geq \dim(\ocompl{V})$,
let
\begin{align}
    \setMv{V}
    &=
    \setm{X \in \Rd[d \times k]}{\image(X) = \ocompl{V}}
    \label{eq:setMvU}
    , \\
    \setMvt{V}
    &=
    \setm{X \in \Rd[d \times k]}{\image(X) + V = \Rd}
    .
    \label{eq:setMvtU}
\end{align}
Note that the value of $k$ is not important as long as it is large enough to allow for the required images,
which is why we do not specify it in the notation.
In statistical applications,
the matrix $X$ could be a data matrix with $k$ observations,
or the basis of the support of a rank-deficient distribution.
The set $\setMv{V}$ consists of matrices whose image is exactly the complement of $V$,
while the set $\setMvt{U}$ consists of any matrix whose image is ``large enough''
to span the whole space together with $U$.
In particular, $\setMvt{U}$ is a superset of $\setMv{V}$ for any $\ocompl{V}$ complementary to $U$.
We have the following relationships between projections, transposes, and pseudoinverses
for matrices in these sets.
\begin{lemma}
    \provenlabel{lemma:projectionMatrices}
    Let $V, W$ be complementary subspaces to $\ocompl{U}$.
    For any $X \in \setMvt{U}$, we have
    $\PUVo X \in \setMv{V}$.
    When restricting the domain and codomain to $\setMv{W}$ and $\setMv{V}$, respectively,
    the mapping $\PUVo$ is a bijection with inverse
    \begin{equation}
        \lr{\PUVo\restrictto{\setMv{W}}}\inv
        =
        \PUWo\restrictto{\setMv{V}}
        .
        \label{eq:projectionRestrictedInverseMatrices}
    \end{equation}
    For $X \in \setMv{U}$, we have
    \begin{align}
        \lr{\PUVo X}\pinv
        &=
        X\pinv \PVVo
        \in
        \setMv{V}\T
        \label{eq:projectionMPInv1}
        , \\
        \lr{\PUVo X}\pinvT
        &=
        \PVVo X\pinvT
        \in
        \setMv{V}
        \label{eq:projectionMPInv1T}
        .
    \end{align}
    For $X \in \setMv{V}$, we have
    \begin{align}
        \lr{\PUUo X}\pinv
        &=
        X\pinv \PUVo
        \in
        \setMv{U}\T
        \label{eq:projectionMPInv2}
        , \\
        \lr{\PUUo X}\pinvT
        &=
        \PVUo X\pinvT
        \in
        \setMv{U}
        \label{eq:projectionMPInv2T}
        .
    \end{align}
\end{lemma}

Again, we illustrate \cref{lemma:projectionMatrices} with a commutative diagram
in \cref{fig:illustrationProjectionMatrices}.
The top row of the diagram is essentially the same as \cref{fig:commutesVectors},
replacing subspaces with the corresponding sets of matrices.
The second row can also be left out by combining the
transposed pseudoinverse into a single mapping.
Notably,
the simple structure of the top row,
where each map is simply a projection along $U$ onto the target image,
is not preserved for the mappings between (transposed) pseudoinverses.
Instead, a more varied set of projections appears,
and mappings between $\setMv{V}$ and $\setMv{W}$ for $V, W \neq U$
do not simplify to single projections at all.
\begin{figure}
    \centering
    \inputtikz{commutesX}
    \caption{%
        Commutative diagram illustrating the relations in \cref{lemma:projectionMatrices}.
        All maps in the top row are projections $\PXYo{U}{S}$,
        with target set $\setMv{S}$.
        Double-headed arrows indicate bijections.
        The transpose $\cdot\T$ and pseudoinverse $\cdot\pinv$
        are self-inverses,
        otherwise the map in each direction is labelled
        on the side of the corresponding arrowhead.
    }
    \label{fig:illustrationProjectionMatrices}
\end{figure}

Using the identities \cref{eq:identityGramProjection,eq:identityGramPseudoinverse},
we can immediately translate the results in \cref{lemma:projectionMatrices}
to the covariance projection $\pUVo$ and its pseudoinverse $\pUVo\pinv$,
yielding \cref{lemma:propertiesSymmetricProjection}.
Again, we illustrate the relationships in the \namecref{lemma:propertiesSymmetricProjection}
with a commutative diagram in \cref{fig:commutesUVpsd}.

\begin{figure}
    \centering
    \inputtikz{commutesUVpsd}
    \caption{%
        Illustration of \cref{lemma:propertiesSymmetricProjection}.
    }
    \label{fig:commutesUVpsd}
\end{figure}

%% file: tex/A_TechnicalDetails.tex
\begin{definition}
    \label{def:MPInverse}
    The \MP{} inverse of a matrix $A \in \Rd[n \times m]$
    is the unique matrix $B \in \Rd[m \times n]$ satisfying the following four equations:
    \begin{align}
        ABA &= A
        \label{eq:MP1}
        \tag{MP1}
        , \\
        BAB &= B
        \label{eq:MP2}
        \tag{MP2}
        , \\
        (AB)\T &= AB
        \label{eq:MP3}
        \tag{MP3}
        , \\
        (BA)\T &= BA
        \label{eq:MP4}
        \tag{MP4}
        .
    \end{align}
    This matrix is denoted by $A\pinv$ and also known as the pseudoinverse of $A$.
\end{definition}

\begin{lemma}
    \provenlabel{lemma:MPInverseCharacterization}
    The following two conditions are sufficient and necessary for $B$ to be the \MP{} inverse of $A$:
    \begin{align*}
        AB
        &=
        \PX{\image(A)}
        , \\
        BA
        &=
        \PX{\image(B)}
        .
    \end{align*}
\end{lemma}
\begin{proof}[\proofref{lemma:MPInverseCharacterization}]
    Since orthogonal projections are symmetric,
    the two conditions imply \cref{eq:MP3,eq:MP4}.
    Furthermore, an orthogonal projection has no effect on vectors from its image,
    so the two conditions also imply \cref{eq:MP1,eq:MP2}.
\end{proof}

%% file: tex/A_Proofs.tex
\subsection{Proofs for \cref{sec:obliqueProjections} and \cref{sec:obliqueProjectionsDetails}}
\label{sec:proofsProjections}

\cref{sec:obliqueProjectionsDetails} provides
additional details and intermediate results
about the oblique projections discussed in \cref{sec:obliqueProjections}.
Below we provide the proofs for both sections together,
ordered such that each proof only relies on results proven before it.

\begin{proof}[\proofref{lemma:projectionFormula}]
    To show that $(B\T A)$ is invertible,
    assume that $\zerovec = (B\T A) z$,
    implying $Az \in \kernel(B\T) = U$.
    Since also $Az \in \image A = \ocompl{V}$,
    we have $Az \in U \cap \ocompl{V} = \set{\zerovec}$,
    by the assumption that $U$ and $\ocompl{V}$ are complementary subspaces.
    Since the columns of $A$ are a basis, this implies $z = \zerovec$,
    and hence that $(B\T A)$ is invertible.

    To show that $\PUVo$ is the oblique projection along $U$ onto $V$,
    let $v = Az \in \image A = \ocompl{V}$ and
    $u \in \kernel(B\T) = U$,
    and verify the conditions from \cref{def:obliqueProjection}:
    \begin{align*}
        \PUVo v
        &=
        A (B\T A)\inv B\T Az
        =
        A z
        =
        v
        , \\
        \PUVo u
        &=
        A (B\T A)\inv B\T u
        =
        \zerovec
        .
    \end{align*}

    The formula for the one-dimensional case
    can be verified by direct computation.
    With $z \in \ocompl{v}$, we have
    \begin{align*}
        \lr{\Id - \frac{uv\T}{v\T u}} u
        &=
        u - u \frac{v\T u}{v\T u}
        =
        \zerovec
        , \qquad
        \lr{\Id - \frac{uv\T}{v\T u}} z
        =
        z - u \frac{v\T z}{v\T u}
        =
        z
        .
    \end{align*}
    Alternatively, one can use the identity $\PUVo = \Id - \PXY{\ocompl{V}}{U}$
    and plug in $u$ and $v$ for $A$ and $B$ above.
\end{proof}

\begin{proof}[\proofref{lemma:projectorMPInv}]
    For the transpose $\PUVo\T$,
    we inherit idempotency from $\PUVo$,
    and the kernel and image satisfy
    \begin{align*}
        \kernel(\PUVo\T)
        &=
        \image(\PUVo)\orth
        =
        V
        , \\
        \image(\PUVo\T)
        &=
        \kernel(\PUVo)\orth
        =
        \ocompl{U}
        ,
    \end{align*}
    implying, $\PUVo\T = \PVUo$.
    We repeatedly use \cref{lemma:pvwpvw} to obtain
    \begin{align*}
        \PUVo
        \lr{
            \PX{\ocompl{U}}
            \PX{\ocompl{V}}
        }
        &=
        \PUVo
        \PXY{U}{\ocompl{U}}
        \PXY{V}{\ocompl{V}}
        =
        \PXY{V}{\ocompl{V}}
        , \\
        \lr{
            \PX{\ocompl{U}}
            \PX{\ocompl{V}}
        }
        \PUVo
        &=
        \PXY{U}{\ocompl{U}}
        \PXY{V}{\ocompl{V}}
        \PUVo
        =
        \PXY{U}{\ocompl{U}}
        ,
    \end{align*}
    which are orthogonal projections onto the image of $\PUVo$ and
    $\PX{\ocompl{U}}\PX{\ocompl{V}}$,
    respectively.
    Hence, by \cref{lemma:MPInverseCharacterization} we have
    $\PUVo\pinv = \PX{\ocompl{U}}\PX{\ocompl{V}}$.
    The identity $\PUVo\pinvT = \PX{U}\PX{V}$ follows directly from the first two.
\end{proof}

\begin{proof}[\proofref{lemma:pvwpvw}]
    For the first equation, let $x \in \Rd$ and set $y = \PUWo x$.
    Then $y \in W\orth$ and $x - y \in U$.
    Since $\PUVo$ vanishes on $U$, we obtain
    \begin{align*}
        \PUVo \PUWo x
        &=
        \PUVo y
        \\
        &=
        \PUVo\lr{y + (x-y)}
        \\
        &=
        \PUVo x
        .
    \end{align*}
    Hence, $\PUVo\PUWo = \PUVo$.
    Taking transposes gives
    \begin{align*}
        \PVUo\PWUo
        &=
        \PWUo
        .
    \end{align*}
    The inverse relationship follows from the fact that $\PUWo$ is the identity on $W\orth$,
    and $\PUVo$ is the identity on $V\orth$.
\end{proof}

\begin{proof}[\proofref{lemma:projectionMatrices}]
    First, we show $\PUVo X \in \setMv{V}$
    for $X \in \setMvt{U}$.
    Using the fact that $U$ is the kernel of $\PUVo$, we have
    \begin{align*}
        \image \PUVo X
        &=
        \PUVo (\image X)
        \\ &=
        \PUVo (\image X + U)
        \\ &=
        \PUVo \Rd
        \\ &=
        \ocompl{V}
        ,
    \end{align*}
    and hence $\PUVo X \in \setMv{V}$.

    To show that $\PUVo\restrictto{\setMv{W}}$ is a bijection,
    note that \cref{lemma:pvwpvw} yields $\PUWo\PUVo = \PUWo$,
    which is the identity on $\setMv{W}$.

    To show \cref{eq:projectionMPInv1},
    let $A = \PUVo X$ and $B = X\pinv \PVVo$.
    From the properties of the \MP{} inverse and $X \in \setMv{U}$,
    we have $XX\pinv = \PUUo$ and $X\pinv = X\pinv X X\pinv = X\pinv \PUUo$.
    For these $A$, $B$ we verify the \MP{} conditions
    by repeatedly applying \cref{lemma:pvwpvw} to yield
    \begin{align*}
        AB
        &=
        \PUVo X X\pinv \PVVo
        \\ &=
        \PUVo \PUUo \PVVo
        \\ &=
        \PVVo
        \\
        BA
        &=
        X\pinv \PVVo \PUVo X
        \\ &=
        X\pinv \PUUo \PVVo \PUVo X
        \\ &=
        X\pinv \PUUo X
        \\ &=
        X\pinv X
        \\
        ABA
        &=
        \PVVo \PUVo X
        =
        \PUVo X
        =
        A
        \\
        BAB
        &=
        X\pinv \PVVo \PVVo
        =
        X\pinv \PVVo
        =
        B
        .
    \end{align*}
    The matrix $(\PUVo X)\pinv$ must be in $\setMv{V}\T$,
    since we have $\PUVo X \in \setMv{V}$ and the transpose of the \MP{}
    has the same image as the original matrix.

    The relationship \cref{eq:projectionMPInv2} follows from similar computations for
    $A = \PUUo X$ and $B = X\pinv \PUVo$.
\end{proof}

\begin{proof}[\proofref{lemma:propertiesSymmetricProjection}]
    To show that $\pUVo(A) \in \setpsd{V}$ for $A \in \setpsdt{U}$,
    we write $A = XX\T$ with $X \in \setMvt{U}$ and
    use the corresponding statement from \cref{lemma:projectionMatrices}
    to obtain $\PUVo X \in \setMv{V}$,
    which implies $\pUVo(A) = \PUVo XX\T \PVUo \in \setpsd{V}$.

    To show \cref{eq:symmetricPvwpvw}, observe
    \begin{align*}
        \pUVo
        \pUWo
        (A)
        &=
        (\PUVo \PUWo) A (\PWUo \PVUo)
        \\ &=
        \PUVo A \PVUo
        \\ &=
        \pUVo(A)
        .
    \end{align*}
    For the bijectivity relationship in \cref{eq:symmetricProjectionRestrictedInverse},
    we write $A \in \setpsd{W}$ as $A = XX\T$ with $X \in \setMv{W}$,
    yielding
    \begin{align*}
        \pUWo \pUVo (A)
        &=
        \pUWo (A)
        =
        \PUWo X \lr{\PUWo X}\T
        =
        XX\T
        =
        A
        .
    \end{align*}
    The other direction and the pseudoinverse relationships follow similarly.
    For \cref{eq:projectionMPInvPSD1}
    consider $A = XX\T \in \setMv{U}$ with $X \in \setMv{U}$,
    yielding
    \begin{align*}
        \pUVo(A)\pinv
        &=
        \lr{(\PUVo X)(\PUVo X)\T}\pinv
        \\ &=
        \lr{\PUVo X}\pinvT \lr{\PUVo X}\pinv
        \\ &=
        \lr{\PVVo X\pinvT} \lr{X\pinv \PVVo}
        \\ &=
        \PVVo X\pinvT X\pinv \PVVo
        \\ &=
        \pVVo(A\pinv)
        .
    \end{align*}
    In the second and the last equality, we use $(XX\T)\pinv = X\pinvT X\pinv$,
    in the third equality we use \cref{eq:projectionMPInv1,eq:projectionMPInv1T}.
    Equation \cref{eq:projectionMPInvPSD2} follows with a similar decomposition,
    using \cref{eq:projectionMPInv2,eq:projectionMPInv2T}.
\end{proof}

\subsection{Proofs for \cref{sec:covarianceStructures}}

\begin{proof}[\proofref{prop:variogramProperties}]
    \cref{prop:variogramProperties1}:
    Note that $\vario(A)$ is always conditionally negative semidefinite for \psd{} matrices $A$,
    the statement only concerns the strictness of this property.
    Let $A \in \setpsdt{\onevec}$,
    meaning $\image A + \onevec\R = \Rd$,
    and let $z \neq \zerovec$ be such that $z\T\onevec = 0$.
    Then
    \begin{align*}
        z\T\vario(A)z
        &=
        z\T\lr{
            \onevec \diag\lr{A}\T + \diag\lr{A}\onevec\T - 2A
        }
        \\ &=
        -2 z\T A z
        .
    \end{align*}
    Since $A$ is positive semidefinite, this term is non-positive.
    To prove strict negativity,
    assume that $z\T A z = 0$.
    Then $Az = \zerovec$, and hence $\image A \subseteq z\orth$.
    Since $z\perp\onevec$, we have $\image A + \onevec\R \subseteq z\orth \subsetneq \Rd$,
    contradicting the assumption that $A \in \setpsdt{\onevec}$.

    To show that $\vario(A) \notin \setcnd$ for $A \in \setpsdany \setminus \setpsdt{\onevec}$,
    consider $A \in \setpsdany$ such that $\image A + \onevec\R \subsetneq \Rd$.
    Then there exists $z \neq \zerovec$ such that $z\T\onevec = 0$ and $Az = \zerovec$,
    yielding $z\T\vario(A)z = 0$.

    \cref{prop:variogramProperties2}:
    For any $A \in \Rdd$ and w.l.o.g. $\onevec\T v = 1$, we have
    \begin{align*}
        \vario\lr{\pevo(A)}
        &=
        \vario\lr{\Pevo A \Pveo}
        \\ &=
        \vario\lr{(\Id - \onevec v\T) A (\Id - v \onevec\T)}
        \\ &=
        \vario(A)
        - \vario(\onevec v\T A)
        - \vario(A v \onevec\T)
        + \vario(\onevec v\T A v \onevec\T)
        \\ &=
        \vario(A)
        - \vario(\onevec (Av)\T)
        - \vario((Av) \onevec\T)
        +  v\T A v \vario(\onevec\onevec\T)
        \\ &=
        \vario(A)
        ,
    \end{align*}
    using $\vario(\onevec\onevec\T) = 0$ and $\vario(\onevec b\T) + \vario(b \onevec\T) = 0$ for any $b \in \Rd$.

    \cref{prop:variogramProperties3}:
    We first show the invariance statement for $A \in \Rdd$:
    \begin{align*}
        -\half \pevo(\vario(A))
        &=
        -\half
        \Pevo
        \lr{
            \onevec \diag\lr{A}\T + \diag\lr{A}\onevec\T - 2A
        }
        \Pveo
        \\ &=
        -\half
        \Pevo
        \lr{
            -2A
        }
        \Pveo
        \\ &=
        \Pevo A \Pveo
        =
        \pevo(A)
        .
    \end{align*}
    The statement about $\Sigma^v \in \setpsd{v}$ then follows,
    since we have $\pevo(\Sigma^v) = \Sigma^v$.

    Next, we show the stated bijection between $\setcnd$ and $\setpsd{v}$.
    We know from \cref{prop:variogramProperties1} and the relation
    $\setpsd{v} \subseteq \setpsdt{\onevec}$ that $\vario(\Sigma^v) \in \setcnd$.
    To show that $-\half\pevo(\Gamma) \in \setpsd{v}$ for $\Gamma \in \setcnd$,
    consider $x \in \Rd$ and observe
    \begin{align*}
        x\T \pevo(-\half\Gamma) x
        &=
        -\half x\T \Pevo \Gamma \Pveo x
        \\ &=
        -\half (\Pveo x)\T \Gamma (\Pveo x)
        ,
    \end{align*}
    which is non-negative since $\Gamma$ is \cnd{}
    and $\onevec \perp \Pveo x$.
    Furthermore, it is zero only if $\Pveo x = \zerovec$,
    which happens if and only if $x \in v\R$.
    Hence, the kernel of the matrix is $v\R$,
    and it is in $\setpsd{v}$.
    Lastly, we use \cref{prop:variogramProperties2} to obtain
    \begin{align*}
        \vario(-\half\pevo(\Gamma))
        &=
        \vario(-\half\Gamma)
        =
        \Gamma
        .
    \end{align*}

    \cref{prop:variogramProperties4}:
    These statements are a direct consequence of \cref{lemma:propertiesSymmetricProjection}
    with $U = \onevec\R$ and $V = v\R$.
\end{proof}

\begin{proof}[\proofref{cor:covarianceStructuresEVT}]
    The stated relationships follow immediately
    when replacing $\peeo$ and $\pexo{\ek}$ by their definition
    from \cref{def:pivw}.
\end{proof}

\begin{proof}[\proofref{lemma:logratiosAsProjections}]
    Let $y$ denote the vector of additive log-ratios $\ALR{j}$,
    extended by $y_j = 0$.
    We have
    \begin{align*}
        y_i
        &=
        \log x_i - \log x_j
        , \quad
        i \in \set{1, \dots, d}
        , \quad
        \text{with}
        \\
        x
        &=
        \cl(w)
        .
    \end{align*}
    Rearranging yields
    \begin{align*}
        y
        &=
        (\log x) - \onevec e_j\T (\log x)
        \\ &=
        \Pexo{e_j} \log x
        \\ &=
        \Pexo{e_j} \log \cl(w)
        \\ &=
        \Pexo{e_j} \lcl(\log w)
        \\ &=
        \Pexo{e_j} \log w
        ,
    \end{align*}
    which is the stated relation for $\ALR{j}$.
    The other relations follow similarly.
\end{proof}

\begin{proof}[\proofref{cor:covarianceStructuresCoDA}]
    For the compositional variation array $\Gamma$
    observe that
    \begin{align*}
        \Gamma_{ij}
        &=
        \Var(\log x_i - \log x_j)
        =
        \Var(\log w_i - \log w_j)
        =
        \Sigma_{ii} + \Sigma_{jj} - 2\Sigma_{ij}
        =
        \vario(\Sigma)_{ij}
        .
    \end{align*}
    The statements about the (centred) log-ratio covariance
    follow from the projection expressions in \cref{lemma:logratiosAsProjections},
    and the identity in \cref{eq:covarianceOfProjection}.
\end{proof}

\subsection{Proofs for \cref{sec:intrinsicGraphicalModels}}

\begin{proof}[\proofref{lemma:logisticNormalALR}]
    By \cref{cor:covarianceStructuresCoDA},
    the covariance matrix of $Y\slr{k} = \ALR{k}(X)$ is obtained
    by removing the $k$-th row and column from
    \begin{equation*}
        \Sigma^{\ek}
        =
        \pexo{\ek}(\Sigma\sonevec)
        .
    \end{equation*}
    By \cref{prop:variogramProperties} and \cref{remark:inverseDroppingZeros},
    the precision matrix $\Theta\slr{k}$ of $Y\slr{k}$ is obtained by
    removing the $k$-th row and column from
    \begin{equation*}
        \Theta^{\ek}
        =
        \pxxo{\ek}(\Theta\sonevec)
        .
    \end{equation*}
    The projection along $\ek$ only changes the values in the $k$-th row and column,
    which are dropped anyway,
    implying that $\Theta\slr{k}_{ij} \Theta\slr{k} = \Theta\sonevec_{ij}$ for all $i,j \neq k$.
    In particular,
    all zeros in $\Theta\sonevec$ are preserved in $\Theta\slr{k}$,
    and $Y\slr{k}$ is a Gaussian graphical model on
    the subgraph of $G$ obtained by removing node $k$ and its incident edges.
\end{proof}

%% file: tex/A_Notation.tex
Throughout the paper,
the notation used for covariance matrices, precision matrices, and variograms
is mostly consistent with literature on extreme value theory,
in particular \citet{engelke2020,engelke2024,hentschel2023,hentschel2025}.
Minor notational differences are
that the matrices here denoted $\Sigma\sonevec$ and $\Theta\sonevec$,
are often shortened to $\Sigma$ and $\Theta$ in the extreme value theory literature,
and that the matrices $\Sigma^{\ek} \in \Rdd$ are often denoted $\tilde\Sigma\slr{k}$.

In compositional data literature, the notation is often different.
While there is no standard notation for covariance matrices, precision matrices, and variograms,
we provide a comparison between our notation
(used e.g. in \cref{sec:backgroundCoDAAitchison,sec:covarianceStructuresCoDA})
and the one used in the standard reference \citet{aitchison_statistical_1986}.
Below, underlined symbols are used to denote the notation from \citet{aitchison_statistical_1986},
while non-underlined symbols follow the notation used in this paper.
\begin{alignat*}{6}
    &\text{Dimension:}
    &\quad
    \ai{D}
    &=
    d
    &&\in \N
    \\
    &\text{Effective dimension:}
    &\quad
    \ai{d}
    &=
    d-1
    &&\in \N
    \\
    &\text{Log-ratio variance:}
    &\quad
    \ai{\tau_{ij}}
    &=
    \Var(x_i/x_j)
    &&\in
    \R
    \\
    &\text{Variation matrix:}
    &\quad
    \mathbf{\ai{T}}
    &=
    \ai{\brackets{\tau_{ij}}}
    =
    \Gamma
    &&\in
    \Rdd
    \\
    &\text{Log-ratio covariance matrix:}
    &\quad
    \mathbf{\ai{\Sigma}}
    &=
    \ai{\brackets{\sigma_{ij}}}
    =
    \Sigma\slr{d}
    &&\in
    \Rdd[(d-1)]
    \\
    &\text{Centred log ratio covariance matrix:}
    &\quad
    \mathbf{\ai{\Gamma}}
    &=
    \ai{\brackets{\gamma_{ij}}}
    =
    \Sigma\sonevec
    &&\in
    \Rdd
    \\
    &\text{Precision matrix:}
    &\quad
    \mathbf{\ai{\Sigma\inv}}
    &=
    \Theta\slr{d}
    &&\in
    \Rdd[(d-1)]
    \\
    &\text{Centred precision matrix:}
    &\quad
    \mathbf{\ai{\Gamma^{-}}}
    &=
    \Theta\sonevec
    &&\in
    \Rdd
    \\
    &\text{Covariance of a basis:}
    &\quad
    \mathbf{\ai{\Omega}}
    &=
    \ai{\brackets{\omega_{ij}}}
    =
    \Sigma
    &&\in
    \Rdd
    \\
    &\text{Centring matrix:}
    &\quad
    \mathbf{\ai{G_D}}
    &=
    \Peeo
    &&\in
    \Rdd
\end{alignat*}